\documentclass[prl,aps,twocolumn,superscriptaddress,dvipsnames,amssymb,9pt,floatfix]{revtex4}
\usepackage{graphicx}
\usepackage{bm}
\usepackage{amsmath}
\usepackage{amssymb}
\usepackage{amsfonts}
\usepackage{fancyhdr}
\usepackage{slashed}
\usepackage{latexsym,bbm}
\usepackage{amsthm}
\usepackage{float}
\usepackage{amsthm}
\usepackage{algorithm,algpseudocode}
\usepackage{enumerate}
\usepackage{natbib}
\usepackage{titletoc}
\usepackage{hyperref}

\newcommand{\ket}[1]{\left |{#1} \right \rangle}

\usepackage{color}
\definecolor{blue}{rgb}{0,0.2,1}

\definecolor{green}{rgb}{0,1,0}

\definecolor{red}{rgb}{0.9,0,0}

\newcommand{\Ord}[1]{\mathcal{O}\left( #1 \right)}
\newcommand{\tOrd}[1]{\tilde{\mathcal{O}}\left( #1 \right)}

\theoremstyle{plain}

\newtheorem{theorem}{Theorem}

\newtheorem{lemma}{Lemma}

\newtheorem{defn}{Definition}

\def\be{\begin{eqnarray}}
\def\ee{\end{eqnarray}}

\usepackage{color}
\definecolor{Pr}{rgb}{0.4,0.3,0.9}

\begin{document}

\title{Provable learning of quantum states with graphical models}

\author{Liming Zhao}\email{zlm@swjtu.edu.cn}
\affiliation{School of Information Science and Technology, Southwest Jiaotong University, Chengdu 610031, China}

\author{Naixu Guo}\email{naixug@u.nus.edu}
\affiliation{Centre for Quantum Technologies, National University of Singapore, 117543, Singapore}

\author{Ming-Xing Luo} 
\affiliation{School of Information Science and Technology, Southwest Jiaotong University, Chengdu 610031, China}
\affiliation{CAS Center for Excellence in Quantum Information and Quantum Physics, Hefei, 230026, China}

\author{Patrick Rebentrost}\email{cqtfpr@nus.edu.sg}
\affiliation{Centre for Quantum Technologies, National University of Singapore, 117543, Singapore}

\begin{abstract} 

The complete learning of an $n$-qubit quantum state requires samples exponentially in $n$. 
Several works consider subclasses of quantum states that can be learned in polynomial sample complexity such as stabilizer states or high-temperature Gibbs states.
Other works consider a weaker sense of learning, such as PAC learning and shadow tomography. 
In this work, we consider learning states that are close to neural network quantum states, which can efficiently be represented by a graphical model called restricted Boltzmann machines (RBMs). 
To this end, we exhibit robustness results for efficient provable two-hop neighborhood learning algorithms for ferromagnetic and locally consistent RBMs.
We consider the $L_p$-norm as a measure of closeness, including both total variation distance and max-norm distance in the limit.
Our results allow certain quantum states to be learned with a sample complexity \textit{exponentially} better than naive tomography. 
We hence provide new classes of efficiently learnable quantum states and apply new strategies to learn them.
\end{abstract}

\maketitle 

\textbf{Introduction.} Fundamental questions in quantum mechanics revolve around probing the nature of quantum systems.  
Quantum tomography investigates the reconstruction of quantum states and quantum channels from measurements on ensembles of identical states \cite{dariano2003quantum, wrightHowLearnQuantum, PhysRevA.77.032322, doi:10.1080/09500349708231894, Flammia_2012, 10.1145/2897518.2897544, 10.1145/2897518.2897585, PhysRevLett.90.193601, anshu2023survey}.  Learning an arbitrary quantum state requires a number of samples exponential in the number of qubits $n$. However, there are subclasses of quantum states that can be learned in $poly(n)$ samples, such as stabilizer states \cite{montanaro2017learning}. Moreover, other than learning all the information about a quantum state, we can learn certain physically useful properties. 
The seminal results on shadow tomography (classical shadows) \cite{aaronson_shadow, huang2022foundations, Huang_2020, Huang_2022, bertoni2023shallow} consider estimating many expectation values of observables. 
Hamiltonian learning considers recovering an unknown Hamiltonian from measurements \cite{PhysRevLett.122.020504, PhysRevLett.112.190501, jianwei_exp_hamiltonian_learning, evans2019scalable,anshu2019hamiltonian}. 

Many works consider the use of machine learning methodologies to investigate quantum systems \cite{mlphase2018, RevModPhys.91.045002}. 
The Restricted Boltzmann Machine (RBM) is a well-studied graphical model with latent (hidden) variables which has been proven to be a universal approximator for arbitrary functions \cite{6796877}. 
At the interplay between RBMs and quantum many-body systems,
Carleo and Troyer represent quantum states with RBMs and proposed the framework called the \textit{Neural-Network Quantum State (NNQ state)} \cite{doi:10.1126/science.aag2302}. They showed that NNQ states can efficiently represent the ground state of the Ising and Heisenberg models.
Certain topological states like toric code states and symmetry-protected
topological cluster state can be represented by a local RBM \cite{deng2017machine}. 
Ref.~\cite{huang2021neural} has proved that any local tensor network can be represented by a local neural network state based on an RBM. 
NNQ states have also been applied to quantum tomography \cite{juan2019gen, torlai2018neural,neugebauer2020neural} and explored in terms of entanglement \cite{PhysRevX.7.021021,levine2019quantum,harney2020entanglement}.

In the past few years, there has been considerable progress in efficient and provable learning algorithms for graphical models. In general, learning all parameters of a graphical model is hard \cite{long2010restricted}. Efficient algorithms can be constructed for special classes of graphical models and 
when considering not learning all parameters but only learning the \textit{structure} of the graphical model, for example.
An unconditional lower bound of $\log n$ holds for the learning of an $n$-spin Ising model \cite{santhanam2012information}.
For RBMs, Bresler \textit{et al.} \cite{bresler2019learning} proposed a classical efficient greedy algorithm for learning the structure in ferromagnetic cases, which can be used to simplify weight learning. 
Goel \cite{goel2019learning} provided a generalization to locally-consistent cases.
A fundamental question arises regarding the connection between these findings and the learning of quantum states. 
In particular, how close do quantum states have to be to classes of RBMs such that the aforementioned algorithms can be effectively employed? 

\begin{figure*}[htpb!]
\includegraphics[width=1\linewidth]{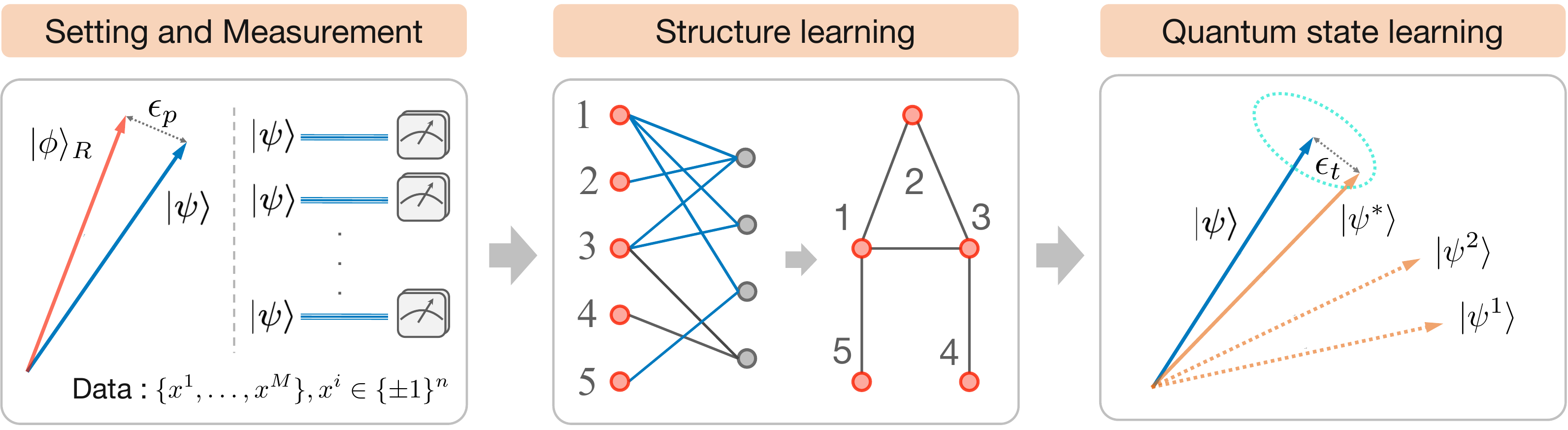}
\caption{\textbf{Provable learning of quantum states with graphical models}. 
(Left panel) Consider an unknown quantum state $\ket{\psi}$.
For the provable statements, we assume it is $\epsilon_p$ close to a neural network quantum state $\ket{\phi}_R$ based on Restricted Boltzmann Machines (RBMs), under the $L_p$ distance defined in Eq.~(\ref{eq_general_error_setting_all_term_bound_main}).
We measure each copy of state $\ket{\psi}$ in the computational basis, obtaining the data $x^i \in \{ \pm1\}^n$.
(Middle panel)
Based on the measurement results, we learn the two-hop neighborhood structure of the RBM corresponding to $\ket{\phi}_R$. An example of an RBM with $5$ visible nodes and $4$ hidden nodes is depicted.  The two-hop edges of node $1$ are shown in blue. The full two-hop neighborhood structure of the visible nodes is also depicted.
(Right panel)
With the knowledge of the structure, we can efficiently perform quantum state learning, here defined as finding parameters to estimate the magnitude of each amplitude of $\ket{\psi}$.  
% With precision $\epsilon_t$ this requires $\text{poly}(n,1/\epsilon_t)$ copies of the state.
We find a sequence of parameters to construct hypothesis states $\{\ket{\psi^1},\cdots\}$, and choose the best one $\ket{\psi^*}$ which is $\epsilon_t$ close to $\ket{\psi}$, with $\text{poly}(n,1/\epsilon_t)$ copies of the state.
As an additional result not shown in this figure, the sample complexity improves further if our focus is solely on estimating the conditional probability of a set of qubits conditioned on the remaining qubits. 
}\label{fig.overall}
\end{figure*}

In this work, we focus on the provable learning of quantum states, see Fig.~\ref{fig.overall} for an overview. We show that for certain classes of quantum states close to an NNQ state based on a locally consistent RBM, the underlying structure of the RBM representation of the NNQ state can be learned efficiently. It allows us to efficiently learn the parameters that can be used to recover the probability distribution in the computational basis of the quantum state.
The technical contribution is to prove the robust structure learning of RBMs. Additionally, we apply the Alphatron algorithm introduced in Ref.~\cite{goel2019learningalphatron} to the task of robust quantum state learning, which leads to the main theoretical guarantees of this work. 
We demonstrate our results through numerical experiments, achieving a   high fidelity of quantum state learning. 

\textbf{Main definitions.}
A Restricted Boltzmann Machine (RBM) is a two-layer (visible layer and hidden layer) neural network without connections among nodes within a layer \cite{salakhutdinov2007restricted,larochelle2008classification}, analogous to a weighted bipartite graph.
An example can be seen in the middle panel of Fig.~\ref{fig.overall}. 
Given an RBM with $n$ visible nodes $\{X_i\}_{i\in [n]}$ and $m$ hidden nodes  $\{Y_k\}_{k\in [m]}$, for any configuration $x\in \{\pm1\}^n$ and $y\in \{\pm1\}^m$, the probability distribution is given by 
\begin{eqnarray}
P(X=x,Y=y)=\frac{1}{Z} \exp (x^TJy+h^Tx+g^Ty),\label{eq_RBM_probability}
\end{eqnarray}
where $Z$ is the partition function, $[n]=\{1,2,\cdots,n\}$, vector $h\in \mathbb{R}^{n}$ and $g\in \mathbb{R}^{m}$ are external fields, and $J\in \mathbb{R}^{n\times m}$ is the interaction matrix. 
An RBM is \textit{locally consistent} if for each hidden node $j\in [m]$, $J_{ij}\geq 0$ or $J_{ij}\leq 0$ for all $i \in [n]$. 

The so-called \textit{two-hop neighborhood} provides information about the structure of an RBM.
We say visible nodes $i$ and $j$ are two-hop neighbors of each other if they are connected via a hidden node in the graph representation.
Let $\mathcal{N}_2(i)$ be the set for two-hop neighbors of visible node $i$.
The two-hop degree of an RBM is defined as  $ d_2:=\max_{i \in[n]}\{|\mathcal N_2(i)|\}$.
We call this interconnected arrangement among visible nodes the structure (underlying graph) of the RBM.
For instance, in the middle panel of Fig.~\ref{fig.overall}, we observe that $\mathcal{N}_2(1)=\{2,3,5\}$, $d_2=3$, and the structure is shown on the right-hand side.

To learn the underlying structure of an RBM, it is necessary to have both lower and upper bounds of the weights. 
We say an RBM is \textit{$(\alpha,\beta)$-non-degenerate} if 
\begin{itemize}
\item for every $i\in [n] $ and $ j\in[m]$,  $|J_{ij}|\geq \alpha$ if $|J_{ij}|\neq 0$;
\item for every $i\in [n]$, $\sum_j|J_{ij}|+|h_i|\leq \beta$;
\item for every $j\in [m]$, $\sum_i|J_{ij}|+|g_j|\leq \beta$. 
\end{itemize}
These assumptions are standard in the literature on learning Ising models and RBMs \cite{bresler2015efficiently,bresler2019learning}. In physics, these bounds imply conditions on the temperature and the characteristic scales of the Hamiltonian of the Gibbs state in Eq.~\ref{eq_RBM_probability}.

\textbf{Results.} We begin with defining classes of quantum states. Neural network quantum states (NNQ states) are quantum states whose amplitudes are described by neural networks. 
It has been shown that any $n$-qubit quantum state 
can be arbitrarily well approximated by an NNQ state based on RBM \cite{huang2021neural}, associating each visible node with a qubit and having potentially an exponential number of hidden nodes.
In this work, we consider NNQ states with real and positive amplitudes. 
\begin{defn}
[Special RBM-NNQ states]\label{definition_quantum_NNQ_class}
Let $\mathcal H_n$ be the Hilbert space of an $n$-qubit system. 
Define 
$\mathcal  C := \mathcal  C(n,d_2,\alpha, \beta)\subset H_n$ as the class of quantum states such that for each $\ket{\phi}_R \in \mathcal  C$ there exists an $(\alpha,\beta)$-nondegenerate locally consistent RBM with two-hop degree $d_2$ with  marginal probability distribution $p(x)$ and $\ket{\phi}_R=\sum_{x\in\{\pm 1\}^n} \sqrt{p(x)}\ket{x}$.
\end{defn}

The main measure of closeness in this work is the 
\textit{$L_p$ distance}. Consider two $n$-qubit quantum states 
$\ket{\psi} = \sum_{x}  \psi(x) \ket{x}$, 
with $\vert \psi (x)\vert^2 =: \tilde p(x)$ and
$\ket{\phi} = \sum_{x}  \phi(x) \ket{x}$, 
with $\vert \phi (x)\vert^2 =: p(x)$. 
Define 
$L_p$ distance between the magnitudes of two states to be
\begin{eqnarray}
L_p(\ket {\psi},\ket{\phi}) :=\left(\textstyle\sum_{x \in\{\pm 1\}^n} \vert  p(x) - \tilde p(x)\vert ^p\right)^{1/p}.
\label{eq_general_error_setting_all_term_bound_main}
\end{eqnarray}
which is the $L_p$ distance of the two vectors $p_1$ and $p_2$.
We define a class of quantum states.
\begin{defn}[$\epsilon_p$-close states]\label{definition_quantum_classes}
Let $H_n$ be Hilbert space of n-qubits. 
Let $\epsilon_p \geq 0$. Define $\mathcal  C(\epsilon_p):=\mathcal{C} (\epsilon_p,n,d_2,\alpha,\beta) \subset H_n$ such that  for each $\ket{\psi} \in \mathcal C(\epsilon_p)$ there exists a $\ket{\phi}_R \in \mathcal C$ with  $L_p$ distance smaller than $\epsilon_p$.
\end{defn}

Even though $\mathcal{C}$ only considers quantum states with real and positive amplitudes, with $\mathcal{C}(\epsilon_p)$ it does allow us to probe quantum states with complex amplitudes.
In the following, we will show that for quantum states belonging to this class, the magnitude of each amplitude is efficiently learnable under certain conditions.
Since $\mathcal{C}$ is contained in $\mathcal{C}(\epsilon_p)$, the results translate to $\mathcal{C}$ obviously.  

\textit{Structure learning}--- The key insight for efficient learning is that the structure of corresponding RBM can tell us certain properties of the unknown quantum states. 
For instance, if a quantum state can be represented by a local RBM, where hidden nodes are only connected to their geometrically local visible nodes, its entanglement entropy (Rényi entropy) satisfies the area law \cite{jia2020entanglement}.  
Based on nonlocal RBM, it can also represent quantum states which exhibit volume-law entanglement \cite{PhysRevX.7.021021}.

Learning the two-hop neighborhood for each node can greatly reduce the sample complexity for finding parameters to recover the magnitudes of the amplitudes. We employ a greedy algorithm \cite{goel2019learning} to learn from the measurements and show that it is suitable to learn quantum states in $\mathcal{C}(\epsilon_p)$ for certain $d_2$ and $\epsilon_p$.
Recall Def. \ref{definition_quantum_NNQ_class} and Def. \ref{definition_quantum_classes}.

\begin{theorem}[Quantum state structure learning, informal]\label{theorem_LRBM_GD_LD_distance}
Let $\epsilon_p \in (0,1)$. Suppose we are given many copies of an $n$-qubit unknown quantum state $\ket{\psi}$ in class $\mathcal C(\epsilon_p)$. Let $\ket{\phi}_R$ in class $\mathcal C$ be the state with $L_p$ distance smaller than $\epsilon_p$. 
If $\epsilon_p$ is small enough (defined by a threshold depending on the parameters), then the two-hop neighborhood of the RBM associated with state $\ket{\phi}_R$ can be estimated with $\tOrd{ \log n}$ copies of $\ket{\psi}$ in time $\tOrd{n^2}$ with a non-zero constant success probability. 
\end{theorem}

The formal version and the proof are given in Theorem \ref{thmMainLC}
in the Appendix.
The algorithm utilizes an observation that if node $u$ and $v$ are two-hop neighbors, their (classical) covariance can be lower bounded by a threshold value \cite{goel2019learning}. 
We show that this value is still valid for quantum states in $\mathcal{C}(\epsilon_p)$.

\textit{Quantum state learning}---With the two-hop structure of the RBM representation learned in the previous step, we can greatly reduce the sample complexity to find parameters to estimate the magnitudes of each amplitude.
If the two-hop degree of the RBM representation of the state is bounded by $d_2$, the magnitudes can be reconstructed using only $\Ord{n 2^{d_2}}$ parameters.  
Otherwise, we need exponentially many (with qubit numbers) parameters in the worst case.

The hardness for the proof is that for quantum states $\ket{\psi}\in \mathcal{C}(\epsilon_p)$, the two-hop degree of the RBM representation of these states may not be bounded by $d_2$, i.e., they may have more complex structures.
We show that if $\epsilon_p$ is small enough, we can still well estimate magnitudes of the states with bounded degree $d_2$.
Also, to learn samples with intrinsic error, i.e., $\epsilon_p$, we use a regression algorithm called \textit{Alphatron}.

Firstly, we show the learning of all parameters. We learn the parameters that fully describe the magnitudes $\tilde{p}(x) = \vert \langle x|\psi\rangle\vert ^2$ of the unknown quantum state, with $\mathcal O\left( {2^{d_2}n^2}\right)$ number of copies of the quantum state. 
Therefore, for $ d_=\mathcal {O}(\text{polylog }n)$, we show a provable efficient learning result for a special class of quantum states.
Recall Def. \ref{definition_quantum_NNQ_class} and Def. \ref{definition_quantum_classes}.
\begin{theorem}[Quantum state learning, informal]\label{Theorem_quantum_state_learning}
Let $\epsilon_p \in (0,1)$.
Suppose we are given many copies of an $n$-qubit unknown quantum state $\ket{\psi}$ in class $\mathcal C(\epsilon_p)$, with magnitudes denoted as $\tilde{p}(x)$. Let $\ket{\phi}_R$ in class $\mathcal C$ be the state with $L_p$ distance smaller than $\epsilon_p$, and the underlying structure of the associated RBM of $\ket{\phi}_R$ is known. 
Let $\epsilon_t>0.$ If the number of copies is $\tOrd{ 2^{d_2}n^2 /\epsilon_t^4}$ and $\epsilon_p$ is small enough, then we find a set of parameters that estimate the magnitudes of the quantum state $\ket{\psi}$, denoted by $\tilde{p}^*(x)$, in time $poly(n)$ with a guarantee of  $\textstyle\sum_x\left\vert \tilde{p}(x) -\tilde{p}^*(x) \right\vert
\in  \Ord{\epsilon_t}$ with high probability. 
\end{theorem}
The formal version of the theorem and the proof are given in Theorem \ref{thmLearnAll} in the Appendix. 
Remarkably, if $d_2 = \log n$, we find that, for $p=1$,  $\epsilon_p=\epsilon_1$ scales as $1/n^2$, while, for $p>1$, $\epsilon_p$  scales as $1/ 2^{n(1-1/p)}$.

Next, we discuss the partial learning of the quantum state in the same setting.  
We are given measurement results of only a subset of the qubits.
Let $p(x)$ be related to $\ket{\phi}_R$ as before.
For a single $u\in[n]$, it has been shown that the conditional probability depends only on the two-hop neighborhood as $p(x_u\vert x_{\neq u}) = p(x_u\vert x_{\mathcal N_2(u) })$ by the Markov property. For a set of nodes $J\subset [n]$, let $\mathcal N_2(J)$ be the set composed of two-hop neighbors of each node in $J$ excluding the nodes in $J$ itself, $\mathcal N_2(J):= \left(\bigcup_{u \in J} \mathcal N_2(u) \right) \setminus J$.  With the complement $x_{\bar{J}}$, we have analogously that $p(x_J\vert x_{\bar J}) = p(x_J\vert x_{\mathcal N_2(J) })$. 
Our result is to provide estimates for the conditional probabilities $ p(x_J\vert x_{\mathcal N_2(J) })$ involving $\Ord{ \vert J \vert 2^{d_2}}$ parameters, which is significantly fewer than the naive representation of a quantum state that would require a number of parameters exponential in $n$. 

\begin{theorem}
[Quantum state partial learning, informal]\label{theorem_conditional_probability_main}
Let $\ket{\psi}$ and $\ket{\phi}_R$ be as in Theorem \ref{Theorem_quantum_state_learning}, and the two-hop neighborhood for all visible nodes of the RBM associated with $\ket{\phi}_R$ be known.
Let $J \subseteq [n]$ and consider its neighborhood $\mathcal N_2(J)$ and the set $J_{\rm tot} = J\cup \mathcal N_2(J)$.
Suppose we are given many copies of the qubits indexed by $J_{\rm tot}$ of the state $\ket{\psi}$. Let $\epsilon_c>0$.
If the number of copies is $\Ord{2^{d_2}\vert J \vert^2/\epsilon_c^4}$, there exists a small number $\epsilon_p\in (0,1)$, such that we can find a set of parameters, that can estimate the conditional probability $ \tilde{p}(x_{J}\vert x_{\mathcal N_2(J)})$ with a bounded error $\epsilon_c$, in time $poly(n).$
\end{theorem}
The formal version of this theorem is presented in Theorem \ref{theorem_appendix_conditional_probability}.
This result is especially useful if the size of the set $J$ is $\log n$, which implies a sample complexity of ${\rm poly}(\log n)$.

\textit{Coherent bit-flips}--- We consider a special case corresponding to coherent bit-flips.  As before,  
$\ket{\phi}_R = \sum_{x \in \{\pm 1\}^n}  \sqrt {p(x)} \ket{x}$, 
and with  $\varrho\in [0,1]$, define another state
as 
\be
\ket{\Psi} &:=&  \sum_{x\in \{\pm 1\}^n} \sqrt{p(x)} ~ {\otimes_{i=1}^n}\left( \sqrt{1-\varrho}\ket{x_i} + \sqrt{\varrho}\ket{- x_i}\right). \nonumber
\label{eq_bitflip_distribution_error_state_main}
\ee
Note that equivalently we have $\ket{\Psi} = \sqrt{1-\epsilon_\varrho} \ket{\phi}_R+ \sqrt{\epsilon_\varrho} \ket{\varphi} $, where $\epsilon_\varrho = 1-(1-\varrho)^n$ and $\ket \varphi$ is an orthogonal state.
This can be understood as a coherent bit-flip error, studied in quantum error correction, or a quantum generalization of random classification noise in quantum PAC learning \cite{10.5555/3291125.3309633} and Huber contamination \cite{NEURIPS2020_bca382c8}.
We also show that the theorems presented in this paper apply to the case that $\varrho$ is small enough. 
Details can be seen in Section \ref{bitflip_appendix} in the Appendix.

\textit{Ferromagnetic RBMs}---We also explicitly consider a class of quantum states close to NNQ states based on ferromagnetic RBMs which is a special case of the locally consistent RBMs, where both pairwise interactions and external fields are non-negative.
We consider the same problem setting by replacing locally consistent RBMs with ferromagnetic RBMs.
The sample complexity for the structure learning is much better on the lower and upper bound strength, i.e., $\alpha$ and $\beta$, achieved with a 
different algorithm \cite{bresler2019learning}.
Once we learn the structure, the sample complexity for quantum state learning is the same with locally consistent cases.
Further details, along with related results and proofs, are provided in Section \ref{structure_learn_FRBM_appendix} in the Appendix.

\textbf{Numerical simulation}. We illustrate our result with a quantum state that is close to an NNQ state based on a locally consistent RBM with a chain underlying structure, where $d_2 = 2$. We consider three cases with varying system sizes, $n=10,15,20$. We set the threshold $\tau=0.05$, each interaction being $1$, and all external fields being $-0.1$.  
We then construct the probability distribution $\tilde p(x)$ from the probability distribution $p(x)$ of the RBM by randomly changing each $p(x)$ within some range (bounded $L_\infty$  distance). 
Then we sample from the distribution $\tilde{p}(x)$ and run the greedy algorithm to obtain the two-hop neighbors. We claim success when we achieve the same underlying structure as that of the RBM. We show the success probabilities along with the different number of samples of the two-hop neighbors learning in panel (a) of Fig.~\ref{figure_success_probability}. With more samples we achieve a higher success probability, for example, for $n=20$, we achieve a success probability $1$ with a sample size of $3700$. For state learning, we use the Alphatron algorithm to obtain the parameters for magnitudes. As a result, we have the fidelities shown in panel (b) of Fig.~\ref{figure_success_probability}. 
For $n=20$, we achieve an average fidelity surpassing $0.99$ for more than $2000$ samples. 
These sample sizes are notably smaller than the exponential $2^{20}$.

\begin{figure}[htpb!]
\includegraphics[width=0.9\linewidth]{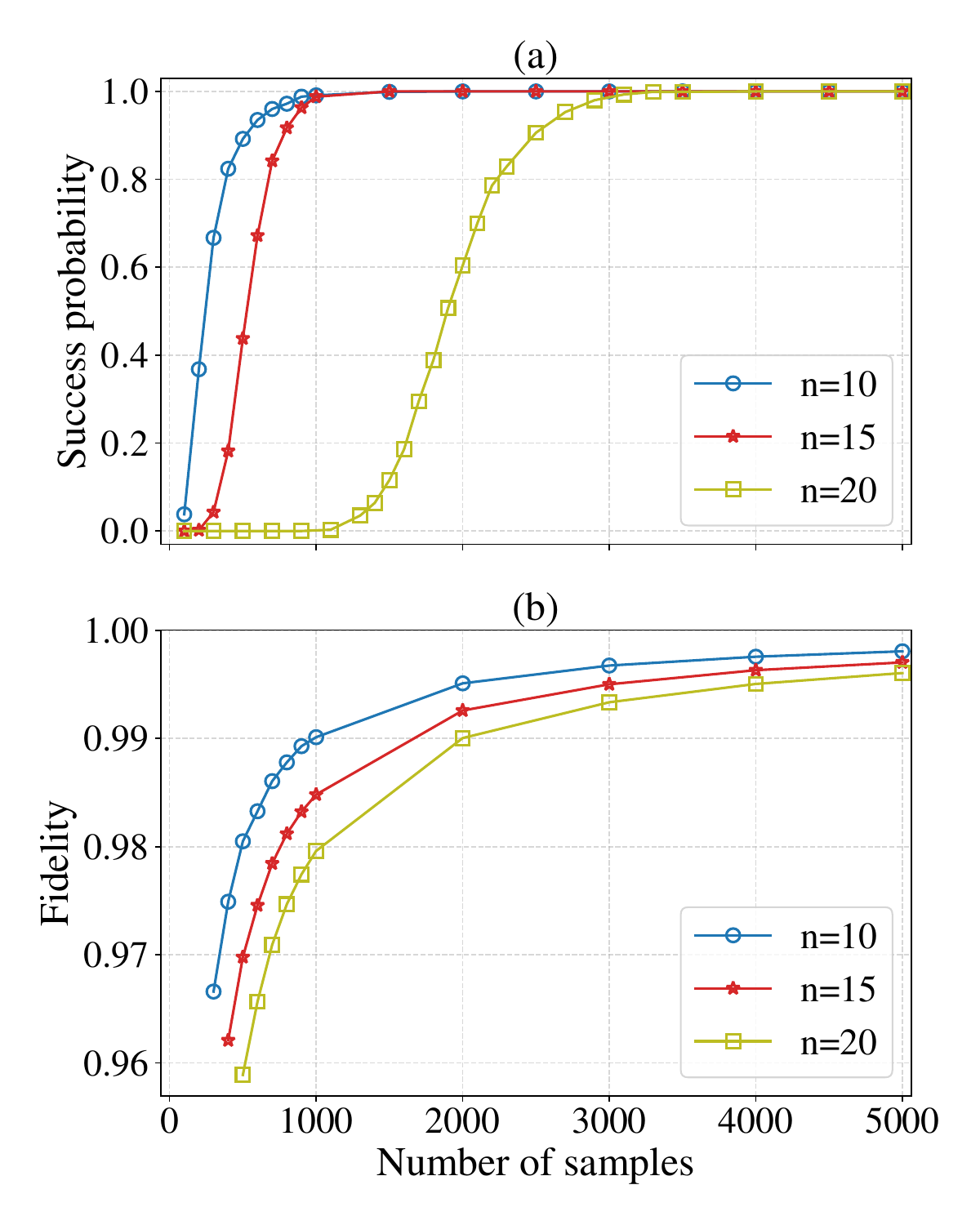}
\caption{Numerical results for learning from measurements of quantum states of a varying number of qubits. (a) The success probability of finding the two-hop neighborhood given a different number of samples. (b) Fidelity of the estimated state from learning from samples. }\label{figure_success_probability}
\end{figure}

\textbf{Conclusion.} We have discussed the learning of quantum states with RBMs using machine learning methods with provable guarantees and near-optimal sample complexity. 
The technical contribution of this work is a robustness result for the RBM learning algorithms. This robustness implies that a much larger class of non-RBM quantum states can be efficiently represented by RBMs with bounded degrees. For such states, the sample complexity for learning them is ${\rm poly}(n)$, in contrast to the $2^n$ dependency of naive tomography. 
Of course, other ways of tomography achieve better sample complexity than naive tomography for certain restricted classes of states, e.g., stabilizer states. We believe our results show alternative ways of quantifying the learnability of quantum states. 
As our methods currently learn the magnitude of the quantum amplitudes, in future work, one can consider an extension to learning the magnitudes and phases of the quantum states. 

\textbf{Acknowledgement} 
This research is supported by the National Natural Science Foundation of China (Grants No. 12204386, No. 62172341), the National Natural Science Foundation of Sichuan Provence (No. 2023NSFSC0447), the Scientific and Technological Innovation Project (No. 2682023CX084), the National Research Foundation, Singapore, and A*STAR under its CQT Bridging Grant and grant NRF2020-NRF-ISF004-3528.
The authors thank Dario Poletti and Feng Pan for their valuable discussions. 
We especially thank Feng Pan for the $n=20$ simulation.

\bibliographystyle{unsrtnat} 
\bibliography{ref}

\appendix
\onecolumngrid
\newpage 
\renewcommand{\appendixname}{Appendix~}
\renewcommand{\thesection}{\Alph{section}}
\renewcommand{\thesubsection}{\arabic{subsection}}
\newpage

\section{Supplementary Material for \\ Provable learning of quantum states with graphical models} 

This supplementary material provides technical details about the theorems discussed in the main paper. It is organized into three main parts, and the overall structure is illustrated in Fig.~\ref{fig.structure}.
The first part is preliminary, where we define the notations and review the necessary concepts related to the Restricted Boltzmann machine (RBM) and the Markov Random Field (MRF). Additionally, we review two greedy algorithms introduced in Ref.~\cite{bresler2019learning,gao2017efficient}, for the convenience of the readers.
The second part is dedicated to robust structure learning. In this part, we establish the robust version of classical learning algorithms and elucidate their application to learning the two-hop neighborhoods of the RBM representation of a quantum state in a certain class from measurement results. The last part is the quantum state learning. With the knowledge of the two-hop neighborhoods learned from measurements, we show that the magnitude of the quantum state can be recovered with a bounded error by using the
Alphatron algorithm \cite{goel2019learningalphatron}.
In particular, we also demonstrate the ability to estimate the distribution of the magnitude of a quantum state that involves a subset of the qubits, conditioned on the configuration we have acquired through measurement results of the remaining qubits.

\begin{figure}[htpb!]
\includegraphics[width=0.7\linewidth]{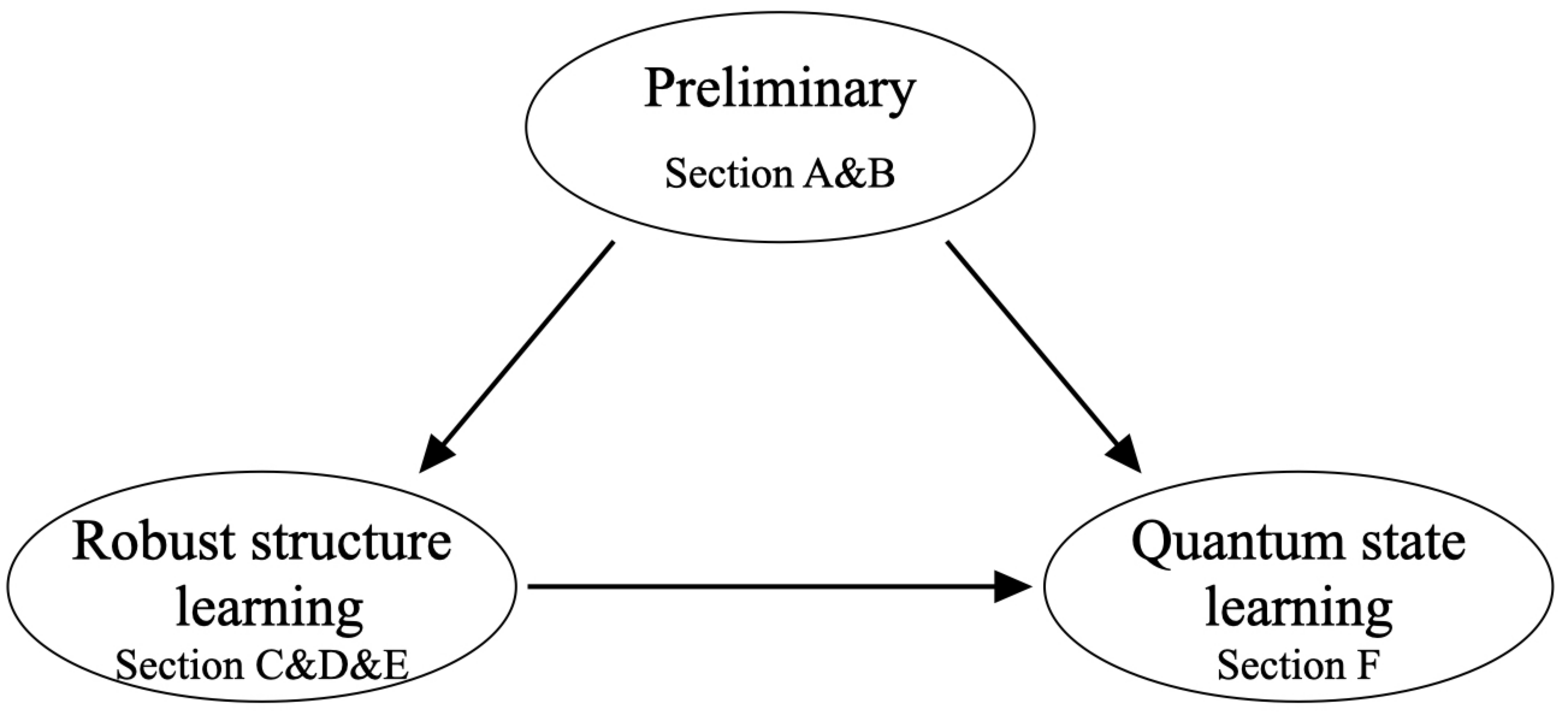}
\caption{Structure for the Appendix} \label{fig.structure}
\end{figure}

\setcounter{secnumdepth}{2}

\section{Preliminaries}

In this section, we first introduce the notations used in this paper.
Then, we review some important properties and results of the Restricted Boltzmann Machine. 
We also introduce two classical greedy algorithms for learning locally consistent and ferromagnetic RBMs for completeness.
Then, we introduce a more general graph model called the Markov random fields (MRF). We review the mapping relation between an RBM and the induced MRF proved in Ref.~\cite{bresler2019learning}. At last, we introduce the Neural Network Quantum (NNQ) state, where the amplitudes are represented by a neural network.

\subsection{Notation}
Let $\mathbb Z_+$ denote the set of positive integers and $[N]=\{1,2,\cdots,N\}$. The set of real numbers is denoted by $\mathbb R$. For a set $I$, $\vert I \vert$ denotes the number of elements in the set.
Let the sigmoid function be $\sigma(x)=1/(1+e^{-x})$. 
For a random variable $X$, the probability for $X=x$ is given by $p(X=x)$.
For simplicity, we sometimes write $p(x)$ with the same meaning.
For a value $a$, we use $\hat{a}$ to represent its empirical version, i.e., estimated from samples.
For a random variable $b$, we use $\mathbb{E}[b]$ to represent its expectation value and use $\hat{\mathbb{E}}[b]$ to represent its empirical expectation value estimated from samples. 

\subsection{Restricted Boltzman Machines}
RBMs are a class of graphical models with latent (hidden) variables that have many applications, such as dimension reduction \cite{doi:10.1126/science.1127647}, feature learning \cite{pmlr-v15-coates11a}, and quantum many-body physics \cite{doi:10.1126/science.aag2302}.
RBMs have also been proven to be universal approximators for arbitrary functions \cite{6796877}.

\begin{defn}[Restricted Boltzmann Machine]
A Restricted Boltzmann Machine $(J,h,g)$ with $n$ visible nodes $\{X_i\}_{i\in [n]}$ and $m$ hidden nodes $\{Y_j\}_{j \in [m]}$ is a weighted bipartite graph and the probability that the model assigns the configuration $(x,y)$, where $x\in \{\pm 1\}^n, y\in \{\pm 1\}^m$, is given by
\begin{eqnarray}
P(X=x,Y=y)=\frac{1}{Z} \exp (x^TJy+h^Tx+g^Ty), \notag
\end{eqnarray}
where $Z$ is the partition function, $J\in \mathbb{R}^{n}\times \mathbb{R}^{m}$ is the interaction matrix, and vector $h\in \mathbb{R}^{n}$, $g\in \mathbb{R}^{m}$ are external fields. 
\end{defn}

Locally consistent and ferromagnetic RBMs are the well-studied subclasses of RBMs, where the definitions are provided as follows.
Notice that a ferromagnetic RBM is a special class of locally-consistent RBM.
\begin{defn}[Locally-Consistent RBM]

An RBM is locally consistent if for each latent node $j\in [m]$, $J_{ij}\geq 0$ or $J_{ij}\leq 0$ for all visible nodes $ I \in [n]$.

\end{defn} 

\begin{defn}[Ferromagnetic RBM]
An RBM is ferromagnetic if the pairwise interactions and external fields are all non-negative i.e. $J_{ij} \geq 0$, $h_i\geq 0$, and $g_j \geq 0$, for all $i \in [n], j \in [m]$.
\end{defn} 

In the following, for simplicity, we use FRBM instead of ferromagnetic RBM.
RBM learning can be divided into structure learning and parameter learning. One purpose of structure learning is to learn the two-hop neighbors of each visible node and use this information for parameter learning. Two visible nodes are two-hop neighbors with each other if they connect to at least one common hidden node. The definition of two-hop neighbors is shown in the following. 
 
\begin{defn}[Two-hop Neighborhood \cite{bresler2019learning}] Let $i \in [n]$ be a visible node for a fixed 
RBM $(J, h, g)$. Denote the two-hop neighborhood of $i$ as $\mathcal N_2(i)$, which is the smallest set 
of visible nodes $S \subset [n] \setminus \{i\}$ such that conditioned on $X_S$,  $X_i$ is conditionally independent of  $X_j$ for all visible nodes $j \in [n]\setminus (S \cup \{i\})$. The  two-hop degree of the RBM is defined as  $ d_2=\max_{i \in[n]}\{|\mathcal N_2(i)|\}$.
\end{defn}

To learn the underlying structure of an RBM, it is necessary to have both lower and upper bounds of the weights. 
We consider the $(\alpha, \beta)$ non-degenerate RBM as follows \cite{bresler2019learning}.
\begin{defn}
An RBM is $(\alpha,\beta)$- non-degenerate if
\begin{itemize}
\item for every $i\in [n] $ and $ j\in[m]$,  $|J_{ij}|\geq \alpha$ if $|J_{ij}|\neq 0$.
\item for every $i\in [n]$, $\sum_j|J_{ij}|+|h_i|\leq \beta$. 
\item for every $j\in [m]$, $\sum_i|J_{ij}|+|g_j|\leq \beta$. 
\end{itemize}
\end{defn}
These assumptions are standard in the literature on learning Ising models and also allow for provable guarantees for the learning of RBMs \cite{bresler2015efficiently}.

\subsection{Classical algorithms for RBM structure learning}

In this section, we review two classical greedy algorithms. 
Bresler, Koehler, and Moitra proposed an algorithm based on influence maximization for learning the two-hop neighborhoods of a visible node for FRBMs \cite{bresler2019learning}. The algorithm takes nearly quadratic time with logarithmic sample complexity with respect to the number of visible nodes. The dependency on the maximum degree and upper bounded strength are single exponential in the algorithm's run-time as well as in optimal sample complexity.
Goel \cite{goel2019learning} on the other hand worked on locally consistent RBMs with arbitrary external fields and proposed an algorithm for learning two-hop neighborhoods on the maximization of conditional covariance which relies on the FKG (Fortuin–Kasteleyn–Ginibre) inequality. The run-time and sample complexity with respect to the number of visible nodes is the same with the FRBM case but the dependency on upper bound strength is doubly exponential.
Note that since an FRBM is also a locally consistent RBM, the algorithm for learning the structure of a locally consistent RBM can be applied to FRBMs as well.

\subsubsection{A greedy algorithm for structure learning of  locally-consistent RBMs}

It has been shown that for any two visible nodes that are the two-hop neighborhood of each other, the conditional covariance can be lower bounded for locally consistent RBMs \cite {goel2019learning}.
Remind that external fields are $h$ and $g$ in Eq.~(\ref{eq_RBM_probability}).
Based on this, the author introduced a classical greedy algorithm that maximizes covariance to learn the two-hop neighborhoods for any visible node of locally consistent RBMs (LC-RBMs) with arbitrary external fields.

\begin{defn}[Conditional covariance]\label{covariance_appen}

The conditional covariance for visible node $u,v$ and a subset $S \subseteq  [n] \setminus  \{u,v\}$   is defined as
\begin{eqnarray}
\text{Cov}(u,v|X_S=x_S)
:= \mathbb{E}[X_uX_v|X_S=x_S] 
- \mathbb{E}[X_u|X_S=x_S ]\mathbb{E}[X_v|X_S=x_S ].
\end{eqnarray}
The average conditional covariance is defined as
\begin{eqnarray}
\textrm{Cov}^{\textrm{avg}}(u,v|S) := \mathbb{E}_{x_S}[\textrm{Cov}(u,v|X_S=x_S) ].  
\end{eqnarray}

\end{defn}

The following property of the average conditional covariance is proved in Ref.~\cite {goel2019learning}.
\begin{lemma}\label{lemma_cov_bound}
Given node $u\in [n]$ and a subset $S \subseteq [n]\setminus \{u\}$ with configuration $X_S$, if node $v\in \mathcal N_2(u)\setminus S$, it satisfies
\begin{eqnarray}
\text{Cov}^{\text{avg}}(u,v|S)\geq \alpha^2 \exp(-12\beta).
\end{eqnarray}
\end{lemma}
Given $M$ samples of visible nodes $\{X^{i}\in \{\pm 1\}^n \}_{i \in [M]}$, the empirical average conditional covariance is defined as
\begin{eqnarray}
\widehat{\text{Cov}}^{\text{avg}}(u,v \vert S) := \widehat{\mathbb{E}}_{x_S}[ \widehat{\text{Cov}}(u,v|X_S=x_S)]. \label{empi_ave_cov} 
\end{eqnarray}
Using Lemma~\ref{lemma_cov_bound}, the author proposed a greedy algorithm \ref{alg_rbm_cov} by maximizing conditional covariance. The following theorem gives the number of samples required and the run time of the algorithm.  
\begin{theorem}[Theorem 2 in Ref.~\cite{ goel2019learning}]\label{Theorem_classical_LRBM}
Given $M_2$ samples of visible nodes $\{X^{i}\in \{\pm 1\}^n\}_{i\in [M_2]}$ of a  $(\alpha,\beta)$-nondegenerate locally-consistent RBM, for $\delta =\frac{1}{2} e^{-2\beta}$,  $\tau= \frac{1}{2}\alpha^2 \exp(-12\beta)$ and $\gamma=\frac{8}{\tau^2}$ , with probability $1-\zeta$, the two-hop neighbours $\mathcal N_2(u)$ of a visible $u$ can be obtained in time $O(M_2n\gamma)$ by using Algorithm \ref{alg_rbm_cov} as long as 
$$M_2 \geq \Omega \left( ( \log(1/ \zeta)+\gamma \log(n))\frac{2^{2\gamma}}{\tau^2\delta^{2\gamma}}  \right).$$
\end{theorem}
Notice that $\gamma \geq d_2,$ then we have $M_2 \geq \Omega \left( ( \log(1/ \zeta)+d_2 \log(n))\frac{2^{2d_2}}{\tau^2\delta^{2d_2}}  \right). $

\begin{figure}[htbp]
\begin{algorithm}[H]
\caption{Greedy algorithm for structure learning of a locally-consistent RBM \cite{ goel2019learning}} \label{alg_rbm_cov}
\begin{algorithmic}[1]
\Require{ Samples $\{X^{i}\}_{i \in [M_2]}$, threshold $\tau$, visible node $u$}
\State Set $S:=\emptyset $ 
\State Let $i^*=\arg\max_v\widehat{\text{Cov}}^{\text{avg}}(u,v\vert S)$, $\eta^*=\max_v\widehat{\text{Cov}}^{\text{avg}}(u,v\vert S)$
\If {$\eta^*\geq \tau$} 
\State $ S=S\cup \{i^*\}$
\Else  
\State go to step 9
\EndIf 
\State Go to step 2
\State pruning step: For each $v\in S$, if $\widehat{\text{Cov}}^{\text{avg}}(u,v\vert S) <\tau$, remove $v$
\Ensure{S}
\end{algorithmic}
\end{algorithm}
\end{figure}

\subsubsection{A greedy algorithm for structure learning of FRBMs}

It has been shown that the two-hop neighbors of an  FRBM can be learned by maximizing the influence function \cite{bresler2019learning}. Given a visible node $u\in [n]$ and a subset $S\subset [n]\setminus \{u\}$  of an FRBM, as in the main paper, the \textit{discrete influence function} is defined as 
\begin{eqnarray}
{I}_u(S) := \mathbb{E}\left [X_u| X_S=\{1\}^{|S|}\right].  \label{eq_definition_influence}
\end{eqnarray}
The discrete influence function is a monotone submodular function for any visible node $u \in [n]$, proved in Ref.~\cite{bresler2019learning}.
The empirical discrete influence function is defined as 
\begin{eqnarray}
\hat{I}_u(S) := \hat{\mathbb{E}}\left[X_u| X_S=\{1\}^{|S|}\right],\label{eq_influence_function}
\end{eqnarray}
where $\hat{\mathbb{E}}$ denotes the empirical expectation based on samples from FRBM. 
Expanding the above equation gives us 
\begin{eqnarray}
\hat{I}_u(S) =   \frac{ \sum_{x_u\in \{\pm 1\}} x_u \hat{p}(X_u=x_u, X_S=\{1\}^{|S|})}{  \hat{p}( X_S=\{1\}^{|S|}) } 
= \frac{2 \hat{p}(X_{S\cup \{u \} }=\{1\}^{|S|+1})}{  \hat{p}( X_S=\{1\}^{|S|}) } -1. \label{influence_expand}
\end{eqnarray}
The distribution probability for $M$ samples  can be obtained by
\begin{eqnarray}
\hat{p}(X=x) =\frac{1}{M}\sum_{i=1}^M \mathbbm{1}_{\{X^{(i)} = x\}}.
\end{eqnarray} 
where $\mathbbm{1}_{\{X^{(i)} = x\}}$ is equal to $1$ when $X^{(i)} = x$, and $0$ otherwise.

The two-hop neighborhoods of each visible node can be found by maximizing the empirical influence function, with the algorithm and theorem provided as follows.

\begin{theorem}[Theorem 6.1 in Ref.~\cite{bresler2019learning}]\label{theorem_rbm_inf}
Given $M_1$ samples $\{X^{i} \in \{\pm1\}^n\}_{i\in[M_1]}$ of the visible nodes of a ferromagnetic Restricted Boltzmann Machine which is $(\alpha,\beta)$ non-degenerate, and has two-hop degree $d_2$. 
For $\delta>0$, if 
\begin{eqnarray}
M_1\geq 2^{2k+3}(d_2/\eta)^2(\log(n)+k\log(en/k))\log(4/\delta),
\end{eqnarray}
where $\eta = \alpha^2 \sigma (-2\beta)(1-\tanh(\beta))^2$, $k=d_2\log(4/\eta)$,  and $\sigma(x)={1}/\left(1+e^{-x}\right)$ for every visible node $u \in [n]$, then Algorithm \ref{algo:greedy} returns $\mathcal{N}_2(u)$ with probability $1-\delta$ in time $\mathcal{O}(M_1kn)$.
\end{theorem}

\begin{figure}[htbp]
\begin{algorithm}[H] 
\caption{Greedy algorithm for structure learning of a FRBM \cite{bresler2019learning}}
\label{algo:greedy}
\begin{algorithmic}[1]
\Require{$M_1$ samples of an RBM.}
\State Set $S_0 \gets \varnothing$.
\For {$t = 0,\cdots,k-1$}
\State{$j_{t+1} \gets \arg \max_{j\in [n]} \widehat{I}_u(S_t \cup \{j \})$.} \label{classical_ferro_step_min_finding}
\State $S_{t+1} \gets S_t \cup \{j_{t+1} \}$.
\EndFor
\State $\widetilde{\mathcal N}_2(u) \gets \{ j\in S_k : \widehat{I} _u(S_k) - \widehat{I}(S_k \setminus \{j\}) \geq \eta \}$
\Ensure {$\widetilde{\mathcal N}_2(u)$}
\end{algorithmic}\label{algGreedy}
\end{algorithm} 
\end{figure}

The total run time for all $n$ visible nodes comes out to be $\mathcal{O}(M_1kn^2)$. 
Step 3 of Algorithm \ref{algo:greedy} is the most time-consuming part which requires $O(M_1n)$. Note that the number of iterations $k$ depends on the two-hop degree and the upper and lower bounds on the strengths of the RBM.

\subsection{Markov Random Fields}\label{MRF.appendix}
An MRF with $n$ variables can be represented by an undirected graph with $ n$ nodes, where the correlations among the variables can be described by the weights of the hyperedges (fully connected subgraph) in the graph.  

\begin{defn}[Markov Random Field]
The probability distribution  on $x\in \{\pm 1\}^n$  of an $r$-wise MRF of $n$ variables can be expressed as 
\begin{eqnarray}
p(X=x)=\frac{1}{Z}\exp(q(x)),
\end{eqnarray}
where $Z$ is the partition function and $q(x) = \sum_{I\subseteq [n]} q_IX_I$ is a multi-linear polynomial referred to as the potential of the MRF, $q_I$ is the coefficient for the subset $I$, and the monomial $X_I=\prod_{i\in I}X_i$.
\end{defn}

Each monomial $x_I$ in the potential $q(x)$ involves a fully-connected subgraph. The fully-connected subgraph means that each pair of nodes in the subgraph $I$ are connected.

RBMs are a subset of MRFs. Moreover, it has been proved that 
the marginal distribution on the visible nodes of an RBM with two-hop degree $d_2$ is a $d_2$-wise Markov Random Field, e.g., used in Ref.~\cite{NIPS2013_7bb06076, bresler2019learning}.
In the other direction, Ref.~\cite{bresler2019learning} proves that every MRF can be converted to an equivalent Restricted Boltzmann Machine. 
A formal statement is provided as follows.

\begin{theorem}[Theorem 4.4 in Ref.~\cite{bresler2019learning}]
Consider an $r$-wise Markov random field of on $\{\pm 1\}^n$. 
Suppose that the degree of the underlying structure of the MRF is $d$ and the coefficients of each monomial of $q(x)$ are bounded by a constant $M$.
Then there is an RBM with $n$ observable nodes and parameters $(J,h,g)$ such that 
\begin{itemize}
\item The induced MRF of the RBM equals the original MRF, i.e., the marginal distribution of visible nodes equals the original MRF.
\item There are at most $O(n2^d)$ hidden nodes.
\item The degree of every hidden node is at most $r$.
\item The two-hop neighborhood of every visible node equals its original MRF neighborhood. The two-hop neighborhood degree $d_2$ equals the degree $d$ of the structure graph of the MRF.
\end{itemize}
\end{theorem}

To learn the underlying structure of an MRF is to learn the neighbors of each node. 
There is an optimal algorithm for learning the underlying graph with $\log n$ samples in time $n^{\Ord{t}}$ and learning the parameters with $n^{\Ord{t}}$  samples  \cite{klivans2017learning}.

\subsection{Neural network quantum states}
As mentioned in the main paper, it has been shown that any $n$-qubit quantum state 
can be arbitrarily well approximated by a neural network quantum state (NNQ state) based on RBM with potentially an exponential number of hidden nodes \cite{huang2021neural}. 
Suppose a 
NNQ state based on RBM is given by
\begin{eqnarray}
\ket{\phi}_R=\sum_{x\in\{\pm 1\}^{n}} {\phi(x)} \ket{x}, ~~~~\vert \phi(x)\vert^2 =  {p(x)}, \label{equ_state_phi_appendix}
\end{eqnarray}
where the magnitude  
$p(x)$ is the marginal probability distribution on the visible nodes of RBMs. For simplicity, in this paper, we focus on NNQ states whose amplitudes are real and positive, i.e., $\phi(x)=\sqrt{p(x)}$.
We define a class named $\mathcal C$ which is defined in the main paper, Def. \ref{definition_quantum_NNQ_class}, and a subset of this class which is named $\mathcal C_F.$ The formal definition is shown in the following.
\begin{defn}[FRBM-NNQ states]
\label{definition_quantum_NNQ_class_appendix}
Let $\mathcal H_n$ be the Hilbert space of an $n$-qubit system.  
Define $\mathcal  C_F := \mathcal C(n,d_2,\alpha,\beta) \subset \mathcal H_n$ the class of quantum states based on $(\alpha,\beta)$-nondegenerate ferromagnetic RBM with two-hop degree $d_2$. 
\end{defn}

With the $L_p$ distance defined in Eq.~(\ref{eq_general_error_setting_all_term_bound_main}), we define the following classes. 
\begin{defn}[$\epsilon_p$-close classes]\label{definition_quantum_classes_appendix}
Let $H_n$ be Hilbert space of n-qubits. 
Let $\epsilon_p \geq 0$. Define two classes as follows: 
\begin{itemize}
\item $\mathcal  C(\epsilon_p)\subset H_n$  such that  for each $\ket{\psi} \in \mathcal C(\epsilon_p)$ there exists a $\ket{\phi}_R \in \mathcal C$ with  $L_p$ distance bounded by $\epsilon_p$.  (Same as Def. \ref{definition_quantum_classes}).
\item $\mathcal  C_F(\epsilon_p)\subset H_n$  such that  for each $\ket{\psi} \in \mathcal C_F(\epsilon_p)$ there exists a $\ket{\phi}_R \in \mathcal C_F$ with  $L_p$ distance bounded by $\epsilon_p$.  
\end{itemize}
We say that a $\ket{\psi}$ is associated with the NNQ state $\ket{\phi}_R$.
\end{defn}
To distinguish with FRBM-NNQ states, sometimes we write LRBM-NNQ state to mean that the NNQ state is based on locally-consistent RBMs, which corresponds to Def.~\ref{definition_quantum_NNQ_class} in the main paper.

\section{Bounds for covariance and influence difference under $L_p$ distance}

In this section, we show that if $\epsilon_p$ is small enough (bounded by $\alpha,\beta$ and $n$), the difference of conditional covariance, defined in Def.~\ref{covariance_appen}, between $\ket{\psi}\in \mathcal{C} (\epsilon_p)$ and the associated NNQ state $\ket{\phi}\in \mathcal{C}$ can be bounded.
Similarly, we show that the difference of the discrete influence functions, defined in Eq.~(\ref{eq_definition_influence}), between $\ket{\psi}\in \mathcal{C}_F(\epsilon_p)$ and the associated FRBM-NNQ state can be bounded if their $L_p$ distance is close enough.
We first prove the following two lemmas which will be used later.

\begin{lemma}\label{lemma_global_distance_subset_bound}
Assume the probabilities $\tilde{p}(x)$ and ${p}(x)$ ($x\in \{\pm 1\}^n$), from the distributions $\widetilde{\mathcal D}$ and ${\mathcal D}$ respectively, are $\epsilon_p$ close under the $L_p$ distance described in Eq.~(\ref{eq_general_error_setting_all_term_bound_main}).
Then for any subset $I\subseteq [n]$ and $x_I\in \{\pm 1\}^{\vert I \vert}$, we have
\begin{eqnarray}
\left\vert p(x_I)- \tilde{ p}(x_I) \right\vert 
\leq 2^{(n-\vert I \vert)(1-1/p)}\epsilon_p.
\end{eqnarray} 
\end{lemma}

\begin{proof}
Let $\bar{I}=[n]\setminus I $ and $x_{\bar{I}}\in\{\pm 1\}^{\vert \bar{I} \vert}$.
Let $P_{x_{I}}$ be a vector of  $p(x_I, x_{\bar{I}})$ over all configurations $x_{\bar{I}}\in \{\pm1\}^{\vert \bar{I}\vert}$.
Note that the dimension of $P_{x_{\bar{I}}}$ is $2^{|\bar{I}|}$.
We define $\tilde{P}_{x_I}$ in a similar way.
By definition of the marginal distribution, we have 
\begin{eqnarray}
\left\vert p(x_I)- \tilde{p}(x_I)\right\vert 
& =& \left\vert \sum_{x_{\bar{I}} }p(x_I,x_{\bar{I}})-\sum_{x_{\bar{I}} }\tilde{p}(x_I,x_{\bar{I}}) \right\vert \nonumber\\
& \leq & \sum_{x_{\bar{I}} } \left\vert p(x_I,x_{\bar{I}})-\tilde{p}(x_I,x_{\bar{I}}) \right\vert  \nonumber \\
&=& \|P_{x_{I}}-\tilde{P}_{x_I}\|_1 \nonumber\\
&\leq& 2^{|\bar I|(1-\frac{1}{p})} \|P_{x_{I}}-\tilde{P}_{x_I}\|_p \nonumber \\
&\leq& 2^{\vert \bar{I}\vert (1-\frac{1}{p})}\epsilon_p = 2^{(n - \vert{I}\vert)(1-\frac{1}{p})}\epsilon_p, \label{eq_ferro_bound_maginal_s_and_u}
\end{eqnarray}
where the first inequality is obtained by the triangle inequality, the second inequality comes from the H\"older's inequality and the last inequality is obtained by the fact that $\|P_{x_{I}}-\tilde{P}_{x_I}\|_p$ is smaller than $\epsilon_p$ since $x_I$ is fixed.
Recall that by H\"older's inequality, for a vector $x\in \mathbb{R}^m$ we have
\begin{eqnarray}
    \|x\|_q\leq m^{1/q-1/p}\|x\|_p.
\end{eqnarray}
In this inequality, set $q=1$ and $m=2^{|\bar{I}|}$, we achieve the inequality above.
\end{proof}

\begin{lemma}\label{lemma_appendix_condi_pro_dis_bound}
If $A, B$ are random variables following a distribution $\mathcal D$, and $C, D $ are random variables following a distribution $\mathcal{\widetilde{D}}$ on the same sample space,  we have 
\begin{eqnarray}
\vert p(A\vert B) - \tilde{p}(C\vert D) \vert \leq \frac{1}{p(B)} \left(\left\vert  p(A,B)-\tilde{p}(C,D)
\right\vert + \left\vert  p(B)-\tilde{p}(D)
\right\vert  \right)
\end{eqnarray}
\end{lemma}
\begin{proof}
We expand the conditional probability by using Bayes' theorem and using an ancillary conditional probability  as following
\begin{eqnarray}
\vert p(A\vert B) - \tilde{p}(C\vert D) \vert  & = &  \left\vert \frac{p(A,B)}{p(B)} -  \frac{\tilde{p}(C,D)}{p(B)} + \frac{\tilde{p}(C,D)}{p(B)}  -\frac{\tilde{p}(C,D)}{\tilde{p}(D)}    \right\vert\nonumber\\
& \leq & 
\left\vert \frac{p(A,B)}{p(B)} -  \frac{\tilde{p}(C,D)}{p(B)} \right\vert + \left\vert \frac{\tilde{p}(C,D)}{p(B)}  -\frac{\tilde{p}(C,D)}{\tilde{p}(D)}    \right\vert\nonumber\\
& = & 
\left\vert \frac{p(A,B)}{p(B)} -  \frac{\tilde{p}(C,D)}{p(B)} 
\right\vert + \frac{\tilde{p}(C,D)}{\tilde{p}(D)}  \left\vert \frac{{p}(B) -\tilde{p}(D) }{p(B)}  \right\vert\nonumber\\
&\leq & 
\frac{1}{p(B)} \left(\left\vert  p(A,B)-\tilde{p}(C,D)
\right\vert + \left\vert  p(B)-\tilde{p}(D)
\right\vert  \right),
\end{eqnarray}
where the first inequality is obtained by using the triangle inequality and the last inequality results from $\frac{\tilde{p}(C,D)}{\tilde{p}(D)}=\tilde{p}(C|D)\leq 1$.
\end{proof}

\subsection{Conditional covariance difference with LRBM-NNQ states}

Assume we are given an $n$-qubit quantum state $\ket{\psi}$ in class $\mathcal C(\epsilon_p)$ as defined in Def.~\ref{definition_quantum_classes_appendix} and $\ket{\phi}_R \in \mathcal C$ is an associate NNQ state. The magnitudes of each amplitude in the computational basis are $\tilde{p}(x)$ and ${p}(x)$ respectively. We show that the difference between the average conditional covariance of any $u,v\in[n]$ conditioned on a $S\subset [n]\setminus \{u,v\}$ of these two probability distributions can be bounded by a function of $\tau,$ if $\epsilon_p$ is small enough. 

\begin{lemma}\label{lemma_local_rbm_global_distance}
If the probabilities $\tilde{p}(x)$ and ${p}(x)$ ($x\in \{\pm\}^n$), from the distributions $\widetilde{\mathcal D}$ and ${\mathcal D}$ respectively, satisfy the $L_p$ distance described in Eq.~(\ref{eq_general_error_setting_all_term_bound_main}), and $\mathcal D$ is a distribution from a locally consistent RBM. Let  ${\text{Cov}}^{\text{avg}}(u,v \vert S)$  and $\widetilde{{\text{Cov}}}^{\text{avg}}(u,v \vert S)$  as average conditional covariance of distributions $\mathcal D$ and $\widetilde{\mathcal D}$ respectively. Let $\tau>0$ and a constant $C>0$, if the following constraint is satisfied
\begin{eqnarray}
\epsilon_p \leq \frac{\tau}{2^{2+n(1-1/p)} C }\label{appendix_eq_error_LRBM} ,
\end{eqnarray}
where $C>0$ is constant,
we have 
\begin{eqnarray}
\bigl \vert {\text{Cov}}^{\text{avg}}(u,v \vert S) - \widetilde{{\text{Cov}}}^{\text{avg}}(u,v \vert S) \bigl\vert  \leq \frac{\tau}{C}.
\end{eqnarray}
\end{lemma}
\begin{proof}
Expand the average conditional covariance defined in Eq.~(\ref{empi_ave_cov})  as the following
\begin{eqnarray}
{\text{Cov}}^{\text{avg}}(u,v \vert S)
&=&{\mathbb{E}}_{x_S}\Bigl[ {\mathbb{E}}[X_uX_v\vert X_S=x_S]-{\mathbb{E}}[X_u\vert X_S=x_S]{\mathbb{E}}[X_v\vert X_S=x_S] \Bigl] 
\nonumber\\
&=&\sum_{x_u,x_v}x_ux_v  {p}(x_u,x_v)
- \sum_{x_u,x_v,x_S}x_ux_v {p}(x_u\vert x_S){p}(x_v\vert x_S)p(x_S) 
\nonumber\\
&=&\sum_{x_u,x_v}x_ux_v {p}(x_u, x_v)
-\sum_{x_u,x_v,x_S}x_ux_v{p}(x_u\vert x_S){p}(x_v, x_S).
\end{eqnarray}

We consider the difference between the average conditional covariance ${\text{Cov}}^{\text{avg}}(u,v \vert S)$  and $ \widetilde{\text{Cov}}^{\text{avg}}(u,v \vert S)$. Let $\mathcal T_1$ and $\mathcal T_2$ represent the first and second terms of the above equation as follows  
\begin{eqnarray}
\mathcal T_1 &=&\sum_{x_u,x_v}x_ux_v {p}(x_u,x_v),\nonumber\\
\mathcal T_2 &=&\sum_{x_u,x_v,x_S}x_ux_v{p}(x_u\vert x_S){p}(x_v, x_S),\label{eq_definition_T2}
\end{eqnarray}
and $\widetilde{\mathcal T}_1$ and $\widetilde{\mathcal T}_2$ represent the corresponding terms in $ \widetilde{\text{Cov}}^{\text{avg}}(u,v \vert S)$. 
We have
\begin{eqnarray}
\left\vert {\text{Cov}}^{\text{avg}}(u,v  \vert S)- \widetilde{\text{Cov}}^{\text{avg}}(u,v \vert S)\right\vert 
& = &\left\vert \mathcal T_1+\mathcal T_2- \widetilde{\mathcal T}_1-\widetilde{\mathcal T}_2\right\vert 
\leq  \left\vert \mathcal T_1-\widetilde{\mathcal T}_1\right\vert
+\left\vert \mathcal T_2- \widetilde{\mathcal T}_2\right\vert \label{equ2}.
\end{eqnarray}
The first term can be divided into two parts: (1) the configurations of node $v$ and $u$ are identical; (2) the configurations of node $v$ and $u$ are different.  We have
\begin{eqnarray}
\left\vert \mathcal T_1-\widetilde{\mathcal T}_1\right\vert 
&=& \left\vert \sum_{x_u,x_v}x_u x_v  p (x_u,x_v)- \sum_{x_u,x_v}x_u x_v \tilde{p}(x_u,x_v) \right\vert  \nonumber\\
&=&\left\vert \sum_{x_u=x_v} \left( p (x_u,x_v)- \tilde{p}(x_u,x_v)\right)  +
\sum_{x_u=-x_v}  \left( p (x_u,x_v)- \tilde{p}(x_u,x_v)\right) \right\vert   \nonumber\\
&\leq &  \left\vert \sum_{x_u=x_v}  p (x_u,x_v)- \tilde{p}(x_u,x_v) \right\vert  + \left\vert  \sum_{x_u=-x_v}   p (x_u,x_v)- \tilde{p}(x_u,x_v)\right\vert,
\end{eqnarray}
where the first inequality comes from the triangle inequality. 
By using Lemma \ref{lemma_global_distance_subset_bound} for set $I = \{u,v\},$ we have
\begin{eqnarray}
\left\vert \mathcal T_1-\widetilde{\mathcal T}_1\right\vert  \leq 2^{n(1-1/p)}\epsilon_p. \label{eq_local_RBM_global_term_1}
\end{eqnarray}

Similarly, we divide the second term in Eq.~(\ref{equ2}) into two parts: whether the configurations of node $u$ and $v$ are the same or not.
For simplicity, we consider the case when $x_u=x_v$.
We have  
\begin{eqnarray}
%1
&& \left\vert \sum_{x_S}  \sum_{x_u=x_v}    \Bigl( p(x_u \vert x_S) p(x_v,x_S)  - \tilde{p}(x_u \vert x_S) \tilde{p}(x_v,x_S)  \Bigl) \right\vert   \nonumber\\
%2
& = &  \left\vert \sum_{x_S}  \sum_{x_u=x_v}   
p(x_u \vert x_S) ( p(x_v,x_S) - \tilde{p}(x_v,x_S)) +  p(x_u \vert x_S)\tilde{p}(x_v,x_S)- \tilde{p}(x_u \vert x_S) \tilde{p}(x_v,x_S) \right\vert  
\nonumber\\
%3
& = &  \left\vert \sum_{x_S}  \sum_{x_u=x_v}   
p(x_u \vert x_S) ( p(x_v,x_S) - \tilde{p}(x_v,x_S)) +  \tilde{p}(x_v,x_S)\left(p(x_u \vert x_S)- \tilde{p}(x_u \vert x_S) \right)  \right\vert  
\nonumber\\
%4
& \leq &   \sum_{x_S}  \sum_{x_u=x_v}   
p(x_u \vert x_S) \bigl\vert p(x_v,x_S) - \tilde{p}(x_v,x_S)  \bigl\vert + \tilde{p}(x_v,x_S)\bigl\vert p(x_u \vert x_S)- \tilde{p}(x_u \vert x_S) \bigl\vert    
\nonumber\\
%5
& \leq &   \sum_{x_S}  \sum_{x_u=x_v}   
\bigl\vert p(x_v,x_S) - \tilde{p}(x_v,x_S)  \bigl\vert + \left(\left\vert p(x_u, x_S)- \tilde{p}(x_u, x_S) \right\vert+ \left\vert p(x_S)- \tilde{p}(x_S) \right\vert\right) \tilde{p}(x_v\vert x_S)   
\nonumber\\
& \leq &   \sum_{x_S}  \sum_{x_u=x_v}   
\left\vert p(x_v,x_S) - \tilde{p}(x_v,x_S)  \right\vert + \left\vert p(x_u, x_S)- \tilde{p}(x_u,  x_S) \right\vert+ \left\vert p(x_S)- \tilde{p}(x_S) \right\vert,    
\end{eqnarray}
where the first inequality is obtained by using the triangle inequality, the second inequality is the result of the fact that the conditional probability is always bounded by $1$ and Lemma \ref{lemma_appendix_condi_pro_dis_bound}.
Combine with the case $x_u=-x_v$,  we have
\begin{eqnarray}
\left\vert \mathcal T_2-\widetilde{\mathcal T}_2\right\vert 
\leq 3\cdot 2^{n(1-1/p)}\epsilon_p  \label{eq_local_RBM_global_term_2}
\end{eqnarray}

Combine all these together, we have
\begin{eqnarray}
\bigl\vert {\text{Cov}}^{\text{avg}}(u,v \vert S) - \widetilde{{\text{Cov}}}^{\text{avg}}(u,v \vert S) \bigl\vert  \leq 2^{n(1-1/p)}\epsilon_p +3\cdot 2^{n(1-1/p)}\epsilon_p \leq 2^{2+n(1-1/p)}\epsilon_p.
\end{eqnarray}
For a constant $C>0$, set 
$ 2^{2+n(1-1/p)}\epsilon_p \leq \frac{\tau}{C}$ and we prove the lemma.
\end{proof}

\subsection{Discrete influence function difference with FRBM-NNQ states}

Similar to the previous subsection, assume we are given an $n$-qubit quantum state $\ket{\psi}$ in class $\mathcal C_F(\epsilon_p)$ as defined and   $\ket{\phi}_R \in \mathcal C_F$ is an associate FRBM-NNQ state, and the magnitudes in the computational basis are $\tilde{p}(x)$ and ${p}(x)$ respectively. We show that the difference between the discrete influence functions of these two probability distributions (magnitudes) can be bounded by $O(\eta)$ if $\epsilon_p$ is small enough. 

\begin{lemma}\label{lemma_FRBM_global distance}
If the probabilities $\tilde{p}(x)$ and ${p}(x)$ ($x\in \{\pm\}^n$), from the distributions $\widetilde{\mathcal D}$ and ${\mathcal D}$ respectively, satisfy the $L_p$ distance described in Eq.~(\ref{eq_general_error_setting_all_term_bound_main}), and $\mathcal D$ is a distribution from an ferromagnetic RBM. Let $\widetilde{{ I}}_u(S\cup \{j\})$  and ${I}_u(S\cup \{j\})$  as the influence function of distributions $\mathcal D$ and $\widetilde{\mathcal D}$ respectively, and $\eta>0$, $\vert S \vert = s\leq k$. For a constant $C>0$, if the following constraint
\begin{eqnarray}
\epsilon_p \leq \frac{\eta}{ 2^{(n-k-1)(1-1/p)+k+3}C} 
\end{eqnarray}
is satisfied, we have
\begin{eqnarray}
\vert {I}_u(S\cup \{j\})-\widetilde{{I}}_{ u}( S\cup \{ j\})\vert \leq \frac{\eta}{C}.\label{eq_lemma_FRBM_global_bound}
\end{eqnarray}
\end{lemma}

\begin{proof}
By the definition of the influence in Eq.~(\ref{eq_definition_influence}), we have  
\begin{eqnarray}
\left\vert {I}_u(S\cup \{j\})-\tilde{{I}}_{ u}( S\cup \{ j\})\right\vert 
&=& \left\vert  {\mathbb{E}}[X_u| X_{S\cup \{j\} }=\{1\}^{s +1}]- \widetilde{{\mathbb{E}}}[ X_{ u}| X_{ S\cup \{ j\}}=\{1\}^{s+1}] \right\vert \nonumber\\
&=& 2\left\vert p(X_u=1| X_{S\cup \{ j\}}=\{1\}^{s+1})- \tilde{p}(X_{u}=1| X_{ S\cup \{ j\} }=\{1\}^{s+1}) \right\vert, \label{eq_influence_difference}
\end{eqnarray}
where $s=\vert S \vert$ is the number of nodes contained in set $S$.   
For the sake of convenience, 
let 
\begin{eqnarray}
a & := & p(X_u=1, X_{S\cup \{ j\}}=\{1\}^{s+1})\nonumber\\
b & := & p( X_{S\cup \{ j\}}=\{1\}^{s+1})\nonumber\\
\tilde{a} & := & \tilde{p}(X_{u}=1, X_{ S\cup \{ j\} }=\{1\}^{s+1}) \nonumber\\
\tilde{b} & := & \tilde{p}( X_{ S\cup \{ j\} }=\{1\}^{s+1}). \label{equ_FRBM_abcd}
\end{eqnarray}
By Eq.~(\ref{eq_influence_difference}) we have 
\begin{eqnarray}
\vert {I}_u(S\cup \{j\})-\tilde{{I}}_{ u}( S\cup \{ j\})\vert 
= 2\left\vert \frac{a}{b}-\frac{ \tilde{a}}{ \tilde{b}} \right\vert  
\leq  \frac{2}{b} \left(\left\vert a- \tilde{a}\right\vert  +  \left\vert {b}-{ \tilde{b}}\right\vert \right)    
\end{eqnarray}
where the first inequality is obtained by using Lemma \ref{lemma_appendix_condi_pro_dis_bound}.
By Lemma \ref{lemma_global_distance_subset_bound}, the distance $\vert {I}_u(S\cup \{j\})-\tilde{{I}}_{ u}( S\cup \{ j\})\vert $ can be upper bounded  by

\begin{eqnarray}
\frac{2}{b} \left(\left\vert a- \tilde{a}\right\vert  +  \left\vert {b}-{ \tilde{b}}\right\vert \right)  
\leq \frac{2(2^{n-s-2})^{(1-1/p)}\epsilon_p +2(2^{n-s-1})^{(1-1/p)}\epsilon_p }{2^{-(s+1)}}   
\leq {(2^{n-s-1})^{(1-1/p)}2^{k+3}}\epsilon_p
\end{eqnarray}
where we use $b\geq 2^{-(s+1)} $ (because in a ferromagnetic model, $x_S=\{1\}^s$ is the most possible configuration to observe for $X_S$), the inequality of the inequality is obtained from the facts that $\vert a - \tilde{a} \vert < (2^{n-s-2})^{(1-1/p)}\epsilon_p, \vert b - \tilde{b} \vert < (2^{n-s-1})^{(1-1/p)}\epsilon_p, b\geq 2^{-(s+1)}.$  Then, for a constant $C\geq 0$, if it satisfies that 
\begin{eqnarray}
{(2^{n-s-1})^{(1-1/p)}2^{k+3}}\epsilon_p \leq \frac{\eta}{C}, 
~ \mbox{i.e.},~  \epsilon_p \leq \frac{\eta}{C (2^{n-s-1})^{(1-1/p)}2^{k+3}} 
\leq \frac{\eta}{C 2^{(n-k-1)(1-1/p)+k+3}}, \label{eq_bound_result_sum_ferromagnetic}
\end{eqnarray}
the Eq.~(\ref{eq_lemma_FRBM_global_bound}) is then satisfied. 
\end{proof}

\section{Bounds for covariance and influence difference under coherent bit-flip distance}\label{bitflip_appendix}

In this section, we consider the case that two quantum states are close under the coherent bit-flip distance. We discuss the difference between the influence functions of the magnitudes 
(average conditional covariances) of a quantum state and an FRBM-NNQ state (LRBM-NNQ state) under this distance. 

Let $\ket{\phi}_R$ be the same as Eq.~(\ref{equ_state_phi_appendix}).
We consider a quantum state $\ket{\Psi}$
\be
\ket{\Psi}=  \textstyle\sum_{x} \sqrt{p(x)}  ~{\otimes_{i=1}^n}\left( \sqrt{\bar{\varrho}}\ket{x_i} + \sqrt{\varrho}\ket{- x_i}\right), \label{eq_bitflip_distribution_error_state}
\ee
where $x \in \{\pm 1\}^n$ and $\bar{\varrho}= 1-\varrho$.
The difference between $\ket{\phi}_R$ and $\ket{\Psi}$ can be understood as a coherent bit-flip error, studied in quantum error correction, or a quantum generalization of random classification noise in quantum PAC learning \cite{10.5555/3291125.3309633} and Huber contamination \cite{NEURIPS2020_bca382c8}.
Assuming that $\ket{\phi}_R$ can be efficiently represented by an RBM, similar to the $L_p$ distance in the previous section, we define two classes of quantum states in the following. 

\begin{defn}[Equivalence class with coherent bit-flip]\label{definition_quantum_classes_bit_flip_appendix}
Let $H_n$ be the Hilbert space of $n$-qubits. 
Let $\varrho \geq 0$. Define two classes as follows: 
\begin{itemize}
\item $\mathcal  C(\varrho)\subset H_n$  such that  for each $\ket{\Psi} \in \mathcal C(\varrho)$ there exists a $\ket{\phi}_R \in \mathcal C$ with one qubit coherent bit flip probability bounded by $\varrho$.  
\item $\mathcal  C_F(\varrho)\subset H_n$  such that  for each $\ket{\Psi} \in \mathcal C_F(\varrho)$ there exists a $\ket{\phi}_R \in \mathcal C_F$ with one qubit coherent bit flip probability bounded by $\varrho$.  
\end{itemize}
We call $\ket{\phi}_R$ the associated NNQ state of state $\ket{\Psi}.$
\end{defn}

If we measure the quantum state  $\ket{\Psi}$ in the computational basis, we obtain $\ket{x_i}$ in probability $1-\varrho$ and $\ket{- x_i}$ with probability $\varrho$.  Let $Y_i$ be the new variable, such that $Y_i=-X_i$ with probability $\varrho$ and  $Y_i=X_i$ with probability $\bar{\varrho}=1-\varrho$. 
We have the following lemma which will be used later.

\begin{lemma}\label{lemma_flip_error_bino_bound_sum}
For a subset $I\subseteq [n]$, let 
$\chi_I :=\sum_{k_I\in K_I}\left(p(Y_I=k_I)- \bar{\varrho}^{\vert I \vert} p(X_I=k_I)\right),$  $k_I\in \{\pm 1\}^{\vert I \vert}$ and $K_I \subseteq \{\pm 1\}^{\vert I \vert},$
we have $\chi_I  \leq  1-\bar{\varrho}^{\vert I \vert}$ and 
\begin{eqnarray}
\left\vert \textstyle\sum_{k_I\in K_I}\left(p[Y_I=k_I] -p(X_I=k_I)\right) \right\vert  \leq   2(1-\bar{\varrho}^{\vert I \vert}).
\end{eqnarray}
\end{lemma}
\begin{proof}
For a vertices subset $I\subseteq [n]$, the  probability of $Y_I$  can be represented by the probability distribution of $X_I$ as follows
\begin{eqnarray}
\sum_{k_I\in K_I} p(Y_I=k_I)= \bar{\varrho}^{\vert I \vert} \sum_{k_I\in K_I}  p(X_I=k_I)+\chi_I  
\end{eqnarray}
We have 
\begin{eqnarray}
\chi_I & = & {\varrho} \bar{\varrho}^{\vert I \vert -1}  \sum_{i\in I} \sum_{k_I\in K_I}  p(X_i=-k_i,X_{I\setminus \{i\}}=k_{I\setminus i})
\nonumber\\
& &+  {\varrho}^2\bar{\varrho}^{ \vert I\vert -2}  \sum_{i,j\in I } \sum_{k_I\in K_I} p\left(X_{\{i,j\}}=\{-k_i,-k_j\},X_{I \setminus \{i,j\}}=k_{I\setminus \{i,j\}}\right) \nonumber\\
& &
+\cdots 
+  {\varrho}^{\vert I\vert }\sum_{k_I\in K_I} p(X_I=-k_I) 
\nonumber\\
&\leq & \sum_{t=1}^{\vert I \vert}  \binom{\vert I\vert}{t} {\varrho}^t \bar{\varrho}^{\vert I\vert-t} 
\nonumber\\
&=& ({\varrho}+(1-{\varrho}))^{\vert I\vert}-\bar{\varrho}^{\vert I\vert} \leq 1-\bar{\varrho}^{\vert I\vert}.
\end{eqnarray}
where the last inequality we use the fact that each $\sum_{k_I\in K_I} p(\cdot)\leq 1$.Then we can obtain
\begin{eqnarray}
\left\vert \sum_{k_I\in K_I}\left(p(Y_I=k_I) -p(X_I=k_I)\right) \right\vert  \leq   (1-\bar{\varrho}^{\vert I \vert})\sum_{k_I\in K_I}p(X_I=k_I) + \vert \chi_I \vert   \leq 2(1-\bar{\varrho}^{\vert I \vert}),
\end{eqnarray}
where the first inequality is obtained by using triangle inequality. 
\end{proof}

\subsection{Conditional covariance distance with LRBM-NNQ states}

Now we turn to the locally consistent RBM. Suppose $\ket{\phi}_R$ is an LRBM-NNQ state. Given an unknown quantum state which is a result of a coherent bit-flip with a probability of $\varrho$ for each qubit of state $\ket{\phi}_R,$ we can bound the difference between the
average conditional covariances of these two states by $O(\tau)$ if $\varrho$ is small enough.

\begin{lemma}\label{lemma_bitflip_LCRBM}

If the probabilities $\tilde{p}(x)$ $(x\in \{\pm\}^n)$ and $\tilde{p}(y)$ from the distributions ${\mathcal D}$  and ${\widetilde{\mathcal D}}$ respectively, ${\mathcal D}$ is from of a locally consistent RBM. For each $i\in[n]$, we have $y_i=-x_i$ with probability $\varrho$ and  $y_i=x_i$ with probability $\bar{\varrho}=1-\varrho$. 
Let  ${\text{Cov}}^{\text{avg}}(u,v \vert S)$  and $\widetilde{{\text{Cov}}}^{\text{avg}}(u,v \vert S)$  as average conditional covariance of distributions $\mathcal D$ and $\widetilde{\mathcal D}$ respectively, and  $\gamma, \tau > 0,$ $\vert S \vert \leq \gamma$. 
 For a constant $C>0$, if it satisfies
\begin{eqnarray}
{\varrho}\leq \frac{\tau}{16C(\gamma +1)},
\end{eqnarray}
  we have $$
\vert {\text{Cov}}^{\text{avg}}(u,v \vert S) - \widetilde{{\text{Cov}}}^{\text{avg}}(u,v \vert S) \vert \leq \frac{\tau}{C}.$$
\end{lemma}

\begin{proof}
Recall that the average conditional covariance in  Eq.~(\ref{equ2}) can be expanded  as follows
\begin{eqnarray}
{\text{Cov}}^{\text{avg}}(u,v \vert S) 
= \sum_{x_u,x_v}x_ux_v {p}(x_u,x_v)
-\sum_{x_u,x_v,x_S}x_ux_v{p}(x_u\vert x_S){p}(x_v, x_S).
~~\label{eq_cov_flip}
\end{eqnarray}
We first consider the difference between the first term corresponding to state $\ket{\phi}_R$ and $\ket{\Psi}$. Let $x_u,x_v\in \{\pm 1 \}$, and ${p}(Y_u=x_u,Y_v=x_v)=\tilde{p}(X_u=x_u, X_v=x_v)= \bar{\varrho}^2 p(X_u=x_u,X_v=x_v)+\chi_{u,v}$,  we have
\begin{eqnarray}
\left\vert \sum_{x_u,x_v}x_u x_v {p}(x_u,x_v)- \sum_{x_u,x_v}x_u x_v \tilde{p}(x_u,x_v) \right\vert  
&=& \left\vert \sum_{x_u=x_v} \left( {p}(x_u,x_v)- \tilde{p}(x_u,x_v)\right) 
\right.
-\left.
\sum_{x_u=-x_v} \left( {p}(x_u,x_v)-\tilde{p}(x_u,x_v)\right)
\right\vert 
\nonumber\\
&= & 2 \left\vert \sum_{x_u=x_v} \left( {p}(x_u,x_v)-\tilde{p}(x_u,x_v)\right) \right\vert  \nonumber\\ 
&\leq & 2 
\sum_{x_u=x_v} (1-\bar{\varrho}^2) {p}(x_u, x_v) +2 \chi_{u,v} \leq 4(1-\bar{\varrho}^2),\label{eq_flip_cov_first_term_bound}
\end{eqnarray}
where the term after the second equal sign is obtained by the fact that $x_u^2=1$, $\sum_{x_u=-x_v}  {p}(x_u, x_v) = 1-\sum_{x_u=x_v}{p}(x_u, x_v) $, and $\sum_{x_u=-x_v}  \tilde{p}(x_u, x_v) = 1-\sum_{x_u=x_v}\tilde{p}(x_u, x_v),$ the last inequality is obtained by using Lemma $\ref{lemma_flip_error_bino_bound_sum}.$

Now turn to the second term of Eq.~(\ref{eq_cov_flip}). For each configuration $x_S\in\{\pm 1\}^{\vert S \vert }$, we first consider the case that $x_u=x_v$.  Let $\tilde{p}(x_u) = p(Y_u=x_u)$, then we can have the same result in Eq.~(\ref{eq_local_RBM_global_term_2})
\begin{eqnarray}
%1
\left\vert \mathcal T_2-\widetilde{\mathcal T}_2\right\vert 
& \leq &   \sum_{x_S}  \sum_{x_u=x_v}   
\left\vert p(x_v,x_S) - \tilde{p}(x_v,x_S)  \right\vert + \left\vert p(x_u, x_S)- \tilde{p}(x_u,  x_S) \right\vert+ \left\vert p(x_S)- \tilde{p}(x_S) \right\vert,\nonumber\\
& \leq & 2(1-\bar{\varrho}^{s+1}) + 2(1-\bar{\varrho}^{s})+2(1-\bar{\varrho}^{s+1}) \leq 6(1-\bar{\varrho}^{s+1})
\end{eqnarray}
where the second inequality is according to Lemma \ref{lemma_flip_error_bino_bound_sum}. 
We can obtain the same bound when $x_u=-x_v$. Combine with the bound in Eq.~(\ref{eq_flip_cov_first_term_bound}) of the first term, it yields that
\begin{eqnarray}
\left\vert {\text{Cov}}^{\text{avg}}(u,v \vert S)-{\widetilde{\text{Cov}}}^{\text{avg}}(u,v \vert S) \right\vert 
&\leq &4(1-\bar{\varrho}^2) + 12(1-\bar{\varrho}^{s+1})
\nonumber\\
&\leq & 16(1-\bar{\varrho}^{s+1})\leq 16(1-\bar{\varrho}^{\gamma +1})\leq 16(\gamma +1)\varrho,
\end{eqnarray}
where the second last inequality uses the fact that $s\leq \gamma$ and the last inequality is obtained by using the union bound.
Let $ 16(\gamma +1)\varrho \leq \frac{\tau}{C},$
we have
\begin{eqnarray}
{\varrho}\leq \frac{\tau}{16C(\gamma +1)}.
\end{eqnarray}

\end{proof}

\subsection{Discrete influence distance with FRBM-NNQ states}
We now discuss the difference of the discrete influence function between an FRBM-NNQ state $\ket{\phi}_R$ and a $\ket{\psi}\in \mathcal{C}(\varrho)$, where the probability distribution (magnitude) denoted as $\mathcal D$ and $\widetilde{\mathcal D}$ respectively. 
Denote $\widetilde{{ I}}_u(S\cup \{j\})$ as the influence function of distribution $\widetilde{\mathcal D}$ and ${I}_u(S\cup \{j\})$ as the influence of istribution ${\mathcal D}$, where the bit flip probability for each qubit is $\varrho$.  We show that the difference of  $\widetilde{{ I}}_u(S\cup \{j\})$  and ${ I}_u(S\cup \{j\})$ can be bounded by $\Ord{\eta}$ if $\varrho$ is small enough. 

\begin{lemma}\label{lemma_FRBM_bitflip_error}
If the probabilities $\tilde{p}(x)$ $(x\in \{\pm\}^n)$ and $\tilde{p}(y)$ from the distributions ${\mathcal D}$  and ${\widetilde{\mathcal D}}$ respectively, ${\mathcal D}$  is from an ferromagnetic RBM. And for each $i\in[n]$, we have $y_i=-x_i$ with probability $\varrho$ and  $y_i=x_i$ with probability $\bar{\varrho}=1-\varrho$. 
Let $\widetilde{{ I}}_u(S\cup \{j\})$  and ${I}_u(S\cup \{j\})$  as the influence function of distributions $\mathcal D$ and $\widetilde{\mathcal D}$ respectively, and $\eta > 0, \vert S \vert \leq k.$ For a constant $C>0$, if 
\begin{eqnarray}
\varrho \leq 
\frac{\eta}{(4 + 2^{k+3})(k+2) C},
\label{eq_probability_bound_FRBM}
\end{eqnarray}
 we have
\begin{eqnarray}
\vert {I}_u(S\cup \{j\})-\widetilde{{I}}_{ u}( S\cup \{ j\})\vert \leq \frac{\eta}{C}.\label{eq_lemma_FRBM_bitflip_bound}
\end{eqnarray}
\end{lemma}

\begin{proof}
For simplicity,  let 
$a, b$ as in Eq.~(\ref{equ_FRBM_abcd}), and 
\begin{eqnarray}
c & := & {p}\left(Y_{u}=1, Y_{ S\cup \{ j\} }=\{1\}^{s+1}\right) \nonumber\\
d & := & {p}\left( Y_{ S\cup \{ j\} }=\{1\}^{s+1}\right) 
\end{eqnarray}
Then according to Lemma  \ref{lemma_flip_error_bino_bound_sum} with just one term in the sum where $K_I = \{1\}^{\vert I \vert }$,  we have 
$c =\bar{\varrho}^{s+2}a+\chi_{S\cup j\cup u},$
$d  = \bar{\varrho}^{s+1}b+\chi_{S\cup j}$
and the bound
\begin{eqnarray}
\chi_{S\cup j\cup u} &\leq & 1-\bar{\varrho}^{s+2},\nonumber\\
\chi_{S\cup j} &\leq & 1-\bar{\varrho}^{s+1}.\label{appendix_eq_flip_bound_FRBM} 
\end{eqnarray}
We now  expand the influence of state $\ket{\phi}_R$ and   $\ket{\psi}$
\begin{eqnarray}
\left\vert {I}_u(S\cup \{j\})-\widetilde{I}_{ u}( S\cup \{ j\})\right\vert 
& =& 2\left\vert {p}\left(X_u=1| X_{S\cup \{ j\}}=\{1\}^{s+1}\right)- {p}\left(Y_{u}=1| Y_{ S\cup \{ j\} }=\{1\}^{s+1}\right) \right\vert. \nonumber\\
& \leq & 2\left\vert \frac{a-c}{b}\right\vert  + 2\left\vert \frac{b-d}{b}\right\vert  \nonumber\\
& \leq &  2 \left\vert \frac{a-\bar{\varrho}^{s+2}a -\chi_{S\cup j \cup u}}{b}  \right\vert  + 2\left\vert \frac{b-\bar{\varrho}^{s+1}b-\chi_{S\cup j}}{b}\right\vert  \nonumber\\
& \leq &  2(2-\bar{\varrho}^{s+2}-\bar{\varrho}^{s+1}) + \left\vert \frac{\chi_{S\cup j \cup u}}{b}  \right\vert  + 2\left\vert \frac{\chi_{S\cup j}}{b}\right\vert  \nonumber\\
& \leq &  2(2-\bar{\varrho}^{s+2}-\bar{\varrho}^{s+1}) + 2(2-\bar{\varrho}^{s+2}-\bar{\varrho}^{s+1}) /{2^{-s-1}} \nonumber\\
& \leq &  2(2-2\bar{\varrho}^{s+2}) + 2^{s+2}(2-2\bar{\varrho}^{s+2})\nonumber\\
& \leq &  4(1-\bar{\varrho}^{k+2}) + 2^{k+3}(1-\bar{\varrho}^{k+2})\nonumber\\
&= & (4 + 2^{k+3})(1-\bar{\varrho}^{k+2}) \leq (4 + 2^{k+3})(k+2){\varrho} \nonumber
\end{eqnarray}
where the first inequality is obtained by using Lemma \ref{lemma_appendix_condi_pro_dis_bound},  the third and forth inequalities we  use the fact that $b\geq 2^{-s-1}$and  Eq.~(\ref{appendix_eq_flip_bound_FRBM}),  the second last inequality we use the fact that $s\leq k$, and the last inequality is obtained by union bound $(1-{\varrho})^{k+2}\geq 1-(k+2)\varrho.$  Let $(4 + 2^{k+3})(k+2){\varrho} \leq \frac{\eta}{C},$  we obtain Eq.~(\ref{eq_probability_bound_FRBM}). 
\end{proof}

\section{Robust structure learning}

In the previous two sections, we have shown that if the $L_p$ distance or the coherent bit-flip probability between two distributions $\mathcal D$ and $\widetilde{\mathcal D}$ is small enough, 
we can bound the covariance distance and the influence function distance between these two distributions when  $\mathcal D$ is from locally consistent and ferromagnetic RBMs respectively.
In this section, we prove the robust structure learning, i.e., we can learn the underlying structure of the RBM representation of state $\ket{\phi}_R$ by giving many copies of $\ket{\psi}\in \mathcal{C}(\epsilon_p)$ or $\ket{\Phi}\in\mathcal{C}(\varrho)$.

\subsection{Robust two-hop neighbors estimation of LRBMs}

Here we consider the case for locally-consistent RBM first. 
Assume we are given many copies of an unknown state $\ket{\psi}\in \mathcal{C}(\epsilon_p)$.  We show that if $\epsilon_p$ is small enough for $\ket{\psi}$ with its associate NNQ state $\ket{\phi}_R\in \mathcal C$, we prove that for a visible $v \in [n]\setminus (S \cup\{u\})$ in the RBM representation of  $\ket{\phi}_R$, $\tau$ can be a  threshold to distinguish whether node $v$ is a two-hop neighborhood of node $u$ or not. 
A similar statement stands for $\ket{\Psi}\in \mathcal{C}(\varrho)$ under the coherent bit-flip distance.

\begin{theorem} \label{thmMainLC}
Let us be given $M_2$ copies of a $n$-qubit unknown quantum state in the union class  $\mathcal C(\epsilon_p)\cup \mathcal C(\varrho)$ defined in Def.~\ref{definition_quantum_classes_appendix} and \ref{definition_quantum_classes_bit_flip_appendix}, $\tau,\gamma$ as in Theorem \ref{Theorem_classical_LRBM}. Denote the average conditional covariance of the RBM representation of the unknown state as $ {\widetilde{\text{Cov}}}^{\text{avg}}(u,v \vert S)$ and the empirical average conditional covariance as  $\widehat{\widetilde{\text{Cov}}}^{\text{avg}}(u,v \vert S)$ for $u,v\in n$ and $S\subset [n]\setminus \{u,v\}.$ 
If $M_2$ satisfies the constraint in Theorem \ref{Theorem_classical_LRBM}, 
 and at least one of the constraints is satisfied, i.e.,  
\begin{eqnarray}
&& (1).  ~ L_p :  \epsilon_p \leq \frac{\tau}{2^{n(1-1/p)+5} };\nonumber
\\
&& (2). \text{~Bit-flip probability}: {\varrho}\leq \frac{\tau}{2^7(\gamma +1)}. 
\label{eq_three_cases_bound_LCRBM}
\end{eqnarray}
 we have
\begin{eqnarray}
\begin{cases}
&\widehat{\widetilde{\text{Cov}}}^{\text{avg}}(u,v \vert S)  >\tau, \text{~for any~} v\in  \mathcal N_2(u),  \\
&\widehat{\widetilde{\text{Cov}}}^{\text{avg}}(u,v \vert S)
<\tau, \text{~for~} v\notin  \mathcal N_2(u).
\end{cases}
\end{eqnarray}
\end{theorem}
\begin{proof}
Given two nodes $u,v$ and a subset $S\subseteq [n]\setminus\{u,v\}$ with $\vert S \vert \leq \gamma$, it has been proved in \cite{gao2017efficient} that if the number of samples satisfies the constraint in Theorem \ref{Theorem_classical_LRBM}, then with probability $1-\rho$, we have
\begin{eqnarray}
\left\vert \widehat{\widetilde{\text{Cov}}}^{\text{avg}}(u,v \vert S) -\widetilde{\text{Cov}}^{\text{avg}}(u,v \vert S)    \right\vert \leq \tau / 2 \label{eq_cov_distance_empirical_true}
\end{eqnarray}

By using  Lemma \ref{lemma_local_rbm_global_distance},  Lemma \ref{lemma_bitflip_LCRBM}, if at least one of the constraints in Eq.~(\ref{eq_three_cases_bound_LCRBM}) satisfied, the difference between the average conditional covariance of the RBM representation of state $\ket{\phi}_R$ (which we denote as $ {\text{Cov}}^{\text{avg}}(u,v \vert S)$ ) and the one of state $\ket{\psi}$ (which we denote as $ {\widetilde{\text{Cov}}}^{\text{avg}}(u,v \vert S)$ ) is bounded as follows
\begin{eqnarray}
\vert {\text{Cov}}^{\text{avg}}(u,v \vert S) - \widetilde{{\text{Cov}}}^{\text{avg}}(u,v \vert S) \vert  \leq \frac{\tau}{8}\label{eq_cov_error_diff}
\end{eqnarray}
Here we set constant $C=8$. 
Combining with Eq.~(\ref{eq_cov_error_diff}) and Eq.~\ref{eq_cov_distance_empirical_true}, we have
\begin{eqnarray}
\left\vert \widehat{\widetilde{\text{Cov}}}^{\text{avg}}(u,v \vert S) -{\text{Cov}}^{\text{avg}}(u,v \vert S)    \right\vert 
& \leq  & \left\vert \widehat{\widetilde{\text{Cov}}}^{\text{avg}}(u,v \vert S) -\widetilde{\text{Cov}}^{\text{avg}}(u,v \vert S)    \right\vert  
+ \left\vert {\widetilde{\text{Cov}}}^{\text{avg}}(u,v \vert S) -{\text{Cov}}^{\text{avg}}(u,v \vert S)    \right\vert  
\nonumber\\
& \leq & \frac{\tau}{2}+\frac{\tau}{8} \label{equ_dis_between_hat_tilde_cov_and_cov}
\end{eqnarray}
From corollary 1 in Ref.~\cite{goel2019learning}, we have
\begin{eqnarray}
\begin{cases} 
{\text{Cov}}^{\text{avg}}(u,v \vert S) \geq 2\tau,~~~&\text{if }  v\in  \mathcal N_2(u) \\
\text{Cov}^{\text{avg}}(u,v \vert S) =0 ,~~~&\text{if } v\notin  \mathcal N_2(u) 
\end{cases}
\end{eqnarray}

Combining with Eq.~(\ref{equ_dis_between_hat_tilde_cov_and_cov}), we have
\begin{eqnarray}
\begin{cases} 
\widehat{\widetilde{\text{Cov}}}^{\text{avg}}(u,v \vert S) \geq 2\tau-\frac{\tau}{8}-\frac{\tau}{2}>\tau,~~~&\text{if }  v\in  \mathcal N_2(u) \\ 
\widehat{\widetilde{\text{Cov}}}^{\text{avg}}(u,v \vert S) \leq \frac{\tau}{8}+\frac{\tau}{2} <\tau  ,~~~&\text{if } v\notin  \mathcal N_2(u).
\end{cases}
\end{eqnarray}
\end{proof}

\subsection{Robust two-hop neighbors estimation of FRBMs }\label{structure_learn_FRBM_appendix}

Assume we are given a quantum state $\ket{\psi}$ close to an FRBM-NNQ state $\ket{{\phi}}_R$. We will show that the difference between the theoretical and empirical influence function of the magnitude probability distribution of state $\ket{\psi}$ can be bounded with enough samples and for any subset $S\subset [n]$ with size $\vert S \vert \leq k$ the following constraint is satisfied
\begin{eqnarray}
\bigl| \tilde{p}(X_S=\{1\}^s)- p(X_S=\{1\}^s)\bigl| \leq \xi,\label{equ_set_S_bound}
\end{eqnarray}
where $ \xi \geq>0,$ $p(\cdot)$ and $\tilde{p}(\cdot)$ are the distribution probability corresponding to $\ket{\phi}_R$ and $\ket{\psi}$ respectively. 
Assuming $\ket{\phi}_R$ is a FRBM-NNQ state, we have the following lemma.
\begin{lemma}\label{lemma_numsamples_FRBM}
Let $\delta,\epsilon,k > 0$, probability $p(x)$ and $\tilde{p}(x)$ is from distribution $\mathcal D$ and $\widetilde{\mathcal{D}}$ of an RBM with $n$ visible noedes respectively, $\mathcal D$ is from an Ferromagnetic RBM,  and Eq.~(\ref{equ_set_S_bound}) is satisfied for any $S\subset [n]$. 
Assume we have many samples from distribution $\widetilde{\mathcal{D}}$, 
to estimate the $\widetilde{I}_u(S)$ with additive precision $\epsilon$, i.e., $\vert \widetilde{I}_u(S)-\widehat{\widetilde{I}}_u(S) \vert \leq \epsilon$ for all $S \subset[n]$ satisfying $|S| = s \leq k$, with probability at least $1-\delta$, 
it suffices to take $T$ many samples with 
\begin{eqnarray}
T \geq \frac{2}{\epsilon^2(2^{-k}-{\xi})^2}\left(\log(n) + k \log\left(\frac{en}{k}\right)\right) \log\left(\frac{4}{\delta}\right).\label{equ_sample_FRBM}
\end{eqnarray}

\end{lemma}

\begin{proof}
The proof is partially based on the previous work \cite{bresler2019learning}. 
Observe that for an FRBM with non-negative external fields,  the configuration $X_S={1}^s$ is the most probable. We have
\begin{eqnarray}
p(X_S=\{1\}^s) \geq 2^{-s}.
\end{eqnarray}
By Eq.~(\ref{equ_set_S_bound}) we obtain
\begin{eqnarray}
\tilde{p}(X_S=\{1\}^s) \geq 2^{-s}-\xi.
\end{eqnarray}

Additionally, notice that the total number of possible sets $S$ is bounded by $\sum_{j=0}^k \binom{n}{j}\leq (en/k)^k.$ 
Now, consider taking $T$ samples.
For each $S$, we define $T_S$ as the number of samples where $X_S=\{1\}^s$. Applying Hoeffding's inequality, we have
\begin{eqnarray}
\tilde{p} (T_S-\mathbb E[T_S]\leq -t)\leq e^{-2t^2/T}.
\end{eqnarray}
Since $\mathbb E[T_S]\geq (2^{-k}-\xi)T $ for $\vert S \vert \leq k$,
\begin{eqnarray}
p\left(T_S< \frac{2^{-k}-\xi}{2} T\right)\leq e^{-\frac{1}{2}T(2^{-k}-\xi)^2}.
\end{eqnarray}

Using the conventional rejection sampling argument,  we see that the samples where $X_S={1}^s$ are independent and identically distributed samples drawn from the conditional law. One way to comprehend this is to regard each sample as generated by first sampling $X_S$ and then subsequently sampling the remaining spins conditioned on $X_S$. Therefore,  through another utilization of Hoeffding's inequality, we can deduce that for a specific selection of $u$ and $S$, we obtain the following
\begin{eqnarray}
p\left(\vert \widetilde I_u(S)-\widehat{\widetilde I}_u(S) \vert\geq \epsilon \vert T_S
\right)\leq 2e^{-2T_S^2\epsilon^2/T_S} \leq 2e^{-2T_S\epsilon^2}.
\end{eqnarray}
Then employing the law of total expectation, we obtain
\begin{eqnarray}
p\left(\vert \widetilde I_u(S)-\widehat{\widetilde I}_u(S) \vert\geq \epsilon\right) 
&=& \mathbb E\biggl[p\left( \left\vert \widetilde I_u(S)-\widehat{\widetilde I}_u(S) \right\vert\geq \epsilon \vert T_S\right)\biggl]\nonumber\\
&\leq& 2 \mathbb E[e^{-2T_S\epsilon^2}]\nonumber\\
&=& 2\mathbb E\biggl[\left(\textbf{1}_{T_S<(2^{-k-1}-\frac{\xi}{2})T}+\textbf{1}_{T_S\geq (2^{-k-1}-\frac{\xi}{2})T}\right)e^{-2T_S\epsilon^2} \biggl]\nonumber\\ 
%&\leq & 2e^{-2T (2^{-k-1}-\frac{\xi}{2})^2}+2e^{-2(2^{-k-1}-\frac{\xi}{2})T\epsilon^2}\nonumber\\ 
&\leq & 4e^{-2T(2^{-k-1}-\frac{\xi}{2})^2\epsilon^2} 
\end{eqnarray}

By the union bound, the probability that $\vert \widetilde{I}_u(S)-\widehat{\widetilde I}_u(S) \vert\geq \epsilon $ for some $u,S$ is at most $n(\frac{en}{k})^k\cdot 4e^{-\frac{1}{2}T (2^{-k}-\xi)^2\epsilon^2}$.
Therefore if we take $T \geq \frac{2}{(2^{-k}-{\xi})^2} (1/\epsilon^2)(\log(n)+k \log(en/k)) \log(4/\delta)$ the result follows.
\end{proof}

We now prove that if given enough samples of a quantum state in the union class $\mathcal C_F(\epsilon_p)\cup C_F(\varrho)$, we can estimate the structure of the state (two-hop neighborhoods of each visible node) if $\epsilon$ or $\varrho$ is small enough. 

\begin{theorem}
Let us be given $T_2$ copies of a quantum state in the union class $\mathcal C_F(\epsilon_p)\cup C_F(\varrho)$ defined in Def.~\ref{definition_quantum_classes_appendix} and 
\ref{definition_quantum_classes_bit_flip_appendix}, $k,\eta,\delta,d_2$ as in Theorem \ref{theorem_rbm_inf}.  Denote the influence function of the RBM representation of the unknown state as $\widetilde{I}_u(S)$ and the empirical influence function as $\widehat{\widetilde{I}}(S)_u$ for $u \in [n]$ and $S\subset [n]\setminus \{u\}.$ 
Assume at least one of the following constraints is satisfied, 
\begin{eqnarray}
&&(1).~ L_p\ \text{distance} : \epsilon_p \leq \frac{\eta}{2^{(n-k-1)(1-1/p)+k+6}};\\
&&(2).~ \text{Coherent bit-flip probability:~} \varrho \leq 
\frac{\eta}{8(4 + 2^{k+3})(k+2)},\label{eq_three_cases_bound}
\end{eqnarray}
and the number of copies satisfies $$T \geq 2^{2k+5}(d_2/\eta)^2(\log(n)+k\log(en/k))\log(4/\delta),$$
we have
\begin{eqnarray}
\vert \widehat{\widetilde{I}}_{ u}(S) - \widehat{\widetilde{I}}_{ u}(S\setminus \{ j\}) \vert
& > & \eta,   \text{~if~} j\in \mathcal N_2(u),\nonumber\\
\vert \widehat{\widetilde{I}}_{ u}(S) -\widehat{\widetilde{I}}_{ u}( S\setminus \{j\}) \vert
& < & \eta, \text{~if~} j\notin\mathcal N_2(u).\label{eq_two_hop_nei_find_FRBM}
\end{eqnarray}

\end{theorem}
\begin{proof}
We first show that if we choose $\xi\leq 2^{-k-1}$ in Lemma \ref{lemma_numsamples_FRBM}, the constraints in Eq.~(\ref{equ_set_S_bound}) can be satisfied under conditions of either the $L_p$ distance or the coherent bit-flip.
For the $L_p$ distance, by Lemma \ref{lemma_FRBM_global distance} and  \ref{lemma_global_distance_subset_bound} and set $C=8$, we have 
$$ \left\vert \widetilde{p}(X_S=\{1\}^s) - {p}(X_S=\{1\}^s) \right\vert \leq 2^{(n-k)(1-{1}/{p})} \frac{\eta}{2^{(n-k-1)(1-{1}/{p})+k+6}}   \leq  2^{-k-1}.$$
We then turn to the coherent bit-flip case. By Lemma \ref{lemma_FRBM_bitflip_error} and also set $C=8$, we have
$$\left\vert \widetilde{p}(X_S=\{1\}^s) - {p}(X_S=\{1\}^s) \right\vert \leq 2(1-\bar{\varrho}^s)\leq 2(1-\bar{\varrho}^{s+1}) \leq \frac{\eta}{2^{s+5}+4} \leq 2^{-k-1}.$$
By using Lemma \ref{lemma_numsamples_FRBM}, and setting $\xi = 2^{-k-1},$  $\epsilon =\frac{\eta}{4d_2}$ in Eq. (\ref{equ_sample_FRBM}), we have 
\begin{eqnarray}
T \geq 2^{2k+5}(d_2/\eta)^2(\log(n)+k\log(en/k))\log(4/\delta).
\end{eqnarray}

Now we turn to prove Eq.~(\ref{eq_two_hop_nei_find_FRBM}).
By using Lemma \ref{lemma_numsamples_FRBM}, as we set $\epsilon =\frac{\eta}{4d_2}$ in Eq. (\ref{equ_sample_FRBM}), we have 
\begin{eqnarray}
\vert \widetilde I_u(S)-\widehat{\widetilde I}_u(S)\vert \leq \eta/(4d_2)\leq \eta/4 .\label{equ_diff_real_empirical}
\end{eqnarray}
From  Lemma 6.3 in Ref.\cite{bresler2019learning}, for any two-hop neighbours $j\in S$, $j\in \mathcal N_2(u)$ we have
\begin{eqnarray}
\vert I_u(S)-I_u(S\setminus \{j\}) \vert \geq 2\eta \label{equ_sj_2eta}.
\end{eqnarray}
If $j$ is not a two-hop neighbourhood of $u$, i.e. $j\notin  \mathcal N_2(u)$, then 
\begin{eqnarray}
\vert I_u(S)-I_u(S\setminus \{j\}) \vert =0. \label{equ_sj_0}
\end{eqnarray}

Then by using Lemma \ref{lemma_FRBM_global distance} and Lemma \ref{lemma_FRBM_bitflip_error}, if at least one
of the constraints in Eq.~(\ref{eq_three_cases_bound}) is satisfied, the difference between the influence of state $\ket{\phi}_R$ as ${I}_u(S\cup \{j\})$ ) and the one of state $\ket{\psi}$ (which we denote as $ \widetilde{I}_u(S\cup \{j\})$ ) is bounded as follows
\begin{eqnarray}
\vert {I}_u(S\cup \{j\})-{\widetilde{I}}_{ u}( S\cup \{ j\})\vert \leq \frac{\eta}{8}.\label{eq_bound_dif_inf_empirical_error}
\end{eqnarray}

By the triangle inequality, we have 
\begin{eqnarray}
%0
\left\vert 
I_u(S)-I_u(S\setminus \{j\}) \right\vert
%1 
& \leq & \left\vert I_u(S)-\widetilde{I}_u(S) \right\vert
%2 
+ \left\vert \widetilde{I}_u(S)- \widehat{\widetilde{I}}_u(S) \right\vert 
%3
+ \left\vert \widehat{\widetilde{I}}_u(S)- \widehat{\widetilde{I}}_u(S\setminus \{j\})\right\vert \nonumber\\ 
& & 
%4
+\left\vert \widehat{\widetilde{I}}_u(S\setminus \{j\}) - \widetilde{I}_u(S\setminus \{j\})\right\vert
%5
+\left\vert  \widetilde{I}_u(S\setminus \{j\})-{I}_u(S\setminus \{j\}) \right\vert.
\end{eqnarray}

Combining Eqs.~(\ref{eq_bound_dif_inf_empirical_error})(\ref{equ_diff_real_empirical}) (\ref{equ_sj_2eta}), 
for any $j\in  \mathcal N_2(u)$ we have
\begin{eqnarray}
\vert {\widehat{\widetilde{I}}}_{ u}(S) - \widehat{\widetilde{I}}_{ u}(S\setminus \{ j\}) \vert
\geq 2\eta-2\cdot\frac{\eta}{4}-2\cdot\frac{\eta}{8}\geq  2\eta-\frac{\eta}{2}-\frac{\eta}{4}>\eta.
\end{eqnarray}

On the other hand, if $j\notin  \mathcal N_2(u)$, combining Eqs.(\ref{eq_bound_dif_inf_empirical_error})(\ref{equ_diff_real_empirical})(\ref{equ_sj_0}), we will obtain
\begin{eqnarray}
\vert \widehat{\widetilde{I}}_{u}( S) -\widehat{\widetilde{I}}_{ u}( S\setminus \{ j\}) \vert
\leq 2 \cdot\frac{\eta}{4}+2\cdot\frac{\eta}{8} < \eta.
\end{eqnarray}
With the threshold $\eta$, we can distinguish whether a node $j$ is a two-hop neighborhood of $u$. 
\end{proof}

Notice that in our robust setting, the number of copies of state required here is slightly bigger than the number of samples required in Ref.~\cite{bresler2019learning}.

\section{Quantum state learning}

In this section, we demonstrate the process of quantum state learning with the estimated two-hop neighborhood structure.
We prove that we can learn magnitudes of each amplitude $\tilde{p}(x)$ of the unknown quantum state $\ket{\psi}$ with good precision. 
The main idea is to learn the parameters of the induced Markov random field (MRF), introduced in Section~\ref{MRF.appendix}.

For a quantum state $\ket{\psi}\in \mathcal{C}(\epsilon_p)$ and its associate NNQ state $\ket{\phi}_R\in \mathcal{C}$, we say that the underlying structure of $\ket{\phi}_R$ is an estimate of the underlying structure of $\ket{\psi}$.
Similar statement stands for $\ket{\Psi}\in\mathcal{C}(\varrho)$ and its associate state $\ket{\phi}_R$.
As described above, the magnitude of a quantum state can be represented by the marginal probability distribution of an RBM. 
In Section~\ref{MRF.appendix}, we mention that for an RBM, there is always an induced MRF, indicating that we can learn the probability distribution of an RBM by learning the induced MRF instead. 
Since we use an estimated structure to approximate the unknown quantum state, there will be some intrinsic error between our estimated result and the real quantum state.
To learn with intrinsic error, we employ an algorithm called Alphatron \cite{goel2019learningalphatron} to find parameters that recover the magnitudes with $ \text{poly} (n)$ samples. Additionally, we can also learn the conditional probability of a set of qubits $J$ conditioned on the remaining qubits, with $\text{poly} (\vert J \vert)$ samples if $d_2$ is equal to $\log \vert J \vert$.

The results are achieved by employing the Alphatron algorithm described in the following.

\begin{theorem}[Alphatron \cite{goel2019learningalphatron}]\label{theorem_alphatron}
Let $\mathcal K$ be a kernel function corresponding to feature map $\phi$ such that for all variable $X'$ in an input domain  $\mathcal X$, $ \Vert \phi(X')\Vert \leq B_1$. Consider samples  $(X'_i,Y'_i)_{i=1}^M$ drawn i.i.d from distribution $\mathcal D$ on $\mathcal X \times [0,1]$ such that $\mathbb E[Y'\vert X'] =\mathcal U (\langle v,\phi(X')\rangle+\epsilon(X'))$ where $u : \mathbb R \to [0,1]$ is a known $L$-Lipschitz non-decreasing function,  $\epsilon : \mathbb R^n\to [-\epsilon_1,\epsilon_1]$ for $\epsilon_1>0$ such that $\mathbb E [\epsilon(X')^2]\leq \epsilon_b$ and $\Vert v \Vert \leq B_2$. Then for $\delta \in (0,1)$, with probability $1-
\delta$, Alphatron with $\lambda = 1/L$, $T=CBL\sqrt{M/\log(1/\delta)}$ and $N=C'M\log(T/\delta)$ for large enough constants $ C, C' >0,$ outputs a hypothesis $h$ such that,
\begin{eqnarray}
\mathbb E_{X',Y'}[\left(h(x)-\mathbb E [Y'\vert X']\right)^2] \leq \Ord{L\sqrt{\epsilon_b} + L \epsilon_1 \sqrt[4]{\frac{\log(1/\delta)}{M}}+ B_1B_2L\sqrt{\frac{\log(1/\delta)}{M}}},
\end{eqnarray}
in time $poly(n, log(1/\delta),t_{\phi}),$  where $t_{\phi}$ is the time required to compute the kernel function $\mathcal K$.
\end{theorem}

Consider an unknown quantum state that $\ket{ \psi }\in \mathcal C(\epsilon_p)$ or $\ket{ \Psi }\in \mathcal C(\varrho)$, and $\ket{\phi}_R$ is the associated NNQ state. Let $q(X)$ and $\tilde{q}(X)$ be the potential of the MRF induced on the observed nodes corresponding to  $\ket{\phi}_R$ and the unknown quantum state respectively, and the partial derivative over node $u$ as  $\tilde{q}_u(X_{\neq u}) =\partial_u \tilde{q}(X)$, ${q}_u(X_{\neq u}) =\partial_u {q}(X).$ Observe that  $ q_u(X)$  contains at most $2^{d_2}$ terms because of the bounded degree constraint, while $\tilde q_u(X_{\neq u})$ contains at most $m_u = \sum_{i=1}^{n-1}\binom{n}{i}$ terms because there are $m_u$ possible monomial containing nodes $u$ for a general MRF. Therefore,  we can divide the partial potential $\tilde q_u(X_{\neq u})$ of the MRF representation of the unknown state into two parts, i.e.,
\begin{eqnarray}
\tilde q_u(X_{\neq u}) =  q_u(X_{\neq u}) + q_\epsilon(X_{\neq u}),\label{eq_q_u_with_error_terms}
\end{eqnarray}
where $ q_u(X_{\neq u})$ contains the monomials from subset in $\mathcal{N}_2(u)$, and we treat $q_\epsilon (X_{\neq u})$ as the error term. 
We show that we can learn the parameters by using the Alphatron algorithm in Ref.~\cite{goel2019learningalphatron} with the knowledge of the underlying graph of the induced MRF, which is the two-hop neighbor structure of the corresponding RBM. 
Particularly, for a node $u\in [n]$, we consider using the Alphatron algorithm to learn the parameters of the partial potential with the knowledge of $\mathcal N_2(u)$.

\begin{lemma}\label{appendix_lemma_partial_potental_bound}
Let $p(X)$ be the marginal probability distribution from a distribution $\mathcal D$, which is over the visible node from an $(\alpha,\beta)$-non-degenerate RBM, $\tilde{p}(X)$ from distribution $\widetilde{\mathcal D}$ as the probabilities distribution which is close to $p(X)$, such that $\vert\mathbb E_{\mathcal D}[{X}_u\vert X_{\neq u}] -  \mathbb E_{\widetilde{\mathcal D}}[ {X}_u\vert X_{\neq u}] \vert \leq \epsilon$  for any $u\in [n]$ and $\epsilon> 0$. Let $q(X)$ and $\tilde{q}(X)$ be the corresponding potentials of the induced MRFs respectively.
For a node $u$, the partial potential  $\tilde{q}_u(X)$ and ${q}_u(X)$ of the MRFs satisfy 
\begin{eqnarray}
\vert \tilde q_u( X_{\neq u}) -q_u( X_{\neq u}) \vert & \leq &  \epsilon', \\
\vert \tilde q_u( X_{\neq u})  \vert &\leq & \beta +  \epsilon',
\end{eqnarray}
where $\epsilon' = \epsilon/(1-(\tanh(\beta)+\epsilon)^2)$.
\end{lemma}
\begin{proof}
Observe that for MRFs with distribution ${p}(x)$, we have $\mathbb E[X_u\vert  X_{\neq u}] 
=  \tanh(q_u( X_{\neq u})) $.
Since $\vert\mathbb E_{\mathcal{D}}[{X}_u\vert X_{\neq u}] -  \mathbb E_{\widetilde{\mathcal D}}[X_u\vert X_{\neq u}] \vert \leq \epsilon$  for $\epsilon>0$, we obtain
\begin{eqnarray}
\vert \tanh(q_u( X_{\neq u})) -\tanh(\tilde{q}_u( X_{\neq u})) \vert 
\leq \epsilon.  
\end{eqnarray}
Now we turn to calculate the bound of the partial potential $\tilde{q}_u( X_{\neq u})$ by using the bound of $\vert q_u(X_{\neq u})\vert$. Using Lemma 6.1 in Ref.~\cite{bresler2019learning}, we have $\vert q_u(X_{\neq u})\vert < \beta$. We now consider the following two cases. First, if $ \vert \tilde q_u(X_{\neq u}) \vert \leq \vert  q_u(X_{\neq u})\vert $, straightforwardly we obtain $\vert \tilde q_u(X_{\neq u}) \vert \leq \beta$.
Therefore, in the following we conside the case that $\vert \tilde q_u(X_{\neq u}) \vert \geq \vert q_u(X_{\neq u})\vert$. 
Note that $\text{arctanh}(x)=\frac{1}{2}\ln \frac{1+x}{1-x}$ is a monotonic increasing function for $x\in (-1,1)$, and the derivative function $\frac{d}{dx}\text{arctanh}(x) = 1/(1-x^2)$.
For $x,y\in (-1,1)$ and $x>y$, we also have
\begin{eqnarray}
\text{arctanh}(x) -\text{arctanh}(y)\leq \frac{1}{1-x^2}(x-y).\label{eq_convex_property}
\end{eqnarray}
With this, we can show that
\begin{eqnarray}
|\tilde q_u( X_{\neq u}) -q_u( X_{\neq u}) | &=&  \vert\text{arctanh}(\tanh(\tilde q_u( X_{\neq u}))) -\text{arctanh}(\tanh(q_u( X_{\neq u}))) \vert \notag\\
&\leq& \max\biggl\{\frac{1}{1-\tanh^2(\tilde q_u( X_{\neq u}))},\frac{1}{1-\tanh^2( q_u( X_{\neq u}))}\biggl\} |\tanh(\tilde q_u( X_{\neq u}))-\tanh(q_u( X_{\neq u}))|\notag \\
&\leq& \frac{\epsilon}{1-(\tanh(\beta)+\epsilon)^2} := \epsilon'.
\end{eqnarray}
% Based on the sign of $q_u(X_{\neq u})$ and $q_u(X_{\neq u})$, 
If $\tilde q_u(X_{\neq u}) q_u(X_{\neq u}) \geq 0$,
by using Eq.~(\ref{eq_convex_property}) with $x =\tanh(\vert \tilde q_u(X_{\neq u})\vert),$ $y= \tanh(\vert q_u(X_{\neq u})\vert),$  we have
\begin{eqnarray}
\vert \tilde q_u(X_{\neq u})\vert -   \vert q_u(X_{\neq u})\vert 
& \leq & \frac{1}{1- \tanh(\vert \tilde q_u(X_{\neq  u})\vert)^2} \bigl(\tanh(\vert \tilde q_u(X_{\neq u})\vert)-\tanh(\vert  q_u(X_{\neq u})\vert)\bigl) \notag\\
&=& \frac{1}{1- \tanh(\vert \tilde q_u(X_{\neq  u})\vert)^2}\left\vert\tanh( \tilde q_u(X_{\neq u}))-\tanh(  q_u(X_{\neq u}))\right\vert\notag \\
&=& \frac{\epsilon}{1-(\tanh(\beta)+\epsilon)^2}=\epsilon'.
\end{eqnarray}
If $\tilde q_u(X_{\neq u}) q_u(X_{\neq u}) \leq 0$,
we have 
\begin{eqnarray}
\tanh(\vert \tilde q_u(X_{\neq u})\vert)-\tanh \left(\vert  q_u(X_{\neq u})\vert\right)
& =  &    \left\vert\tanh( \tilde q_u(X_{\neq u}))+\tanh(  q_u(X_{\neq u}))\right\vert \nonumber\\
&\leq & \left\vert\tanh( \tilde q_u(X_{\neq u}))-\tanh(q_u(X_{\neq u}))\right\vert \leq \epsilon,
\end{eqnarray}
where the first inequality arises from the property that $\tanh(\cdot)$ is an odd function.
Combining these two cases together we obtain
\begin{eqnarray}
\vert \tilde q_u(X_{\neq u})\vert  
& \leq &  \frac{\epsilon}{1-  \tanh^2(\vert \tilde q_u(X_{\neq u})\vert)}  + \vert q_u(X_{\neq u})\vert\nonumber\\
& \leq &  \frac{\epsilon }{1-(\tanh(\beta)+\epsilon)^2}  + \beta =\epsilon'+\beta.
\end{eqnarray}

\end{proof}

\begin{lemma}\label{lemma_appendix_conditional_probability_bound} 
Let $\tilde{p}(x)$ be the probability distribution of state $\ket{\psi} \in \mathcal C (\epsilon_p)$ or state $\ket{\Psi} \in \mathcal C (\varrho)$ and $p(x)$ is the distribution probability of an associate NNQ state $\ket{\phi}_R$.  Suppose $p(x)$ and $\tilde{p}(x)$ satisfies the setting in Lemma \ref{appendix_lemma_partial_potental_bound}, for a node $u\in [n]$,   the following constraint (1) is satisfied for state $\ket{\psi}$ and  constraint (2) is satisfied for state $\ket{\Psi}$, 
\begin{eqnarray}
& (1).& ~~ \vert\mathbb E_{\mathcal D}[{X}_u\vert X_{\neq u} = x_{\neq u}  ] -  \mathbb E_{\widetilde{\mathcal D}}[X_u\vert  X_{\neq u}= x_{\neq u}] \vert  \leq \frac{2^{(n-1)(1-1/p )+2}\epsilon_p}{\sigma(-2\beta)^{d_2}}  \nonumber\\
& (2).&  ~~ \vert\mathbb E_{\mathcal{D}}[{X}_u\vert X_{\neq u} = x_{\neq u}  ] -  \mathbb E_{\widetilde{\mathcal D}}[X_u\vert X_{\neq u}=x_{\neq u}] \vert  \leq \frac{8n\varrho}{\sigma(-2\beta)^{d_2}}.
\label{eq_appendix_conditional_probability_bound}
\end{eqnarray}   
\end{lemma}
\begin{proof}
By using Lemma \ref{lemma_appendix_condi_pro_dis_bound}, we have 
\begin{eqnarray}
& & \vert\mathbb E_{\mathcal{D}}[{X}_u\vert X_{\neq u} = x_{\neq u}] -  \mathbb E_{\widetilde{\mathcal D}}[X_u\vert X_{\neq u} = x_{\neq u}] \vert \nonumber\\
& = & 2\left\vert p({X_u}=1 \vert  x_{\neq u})- \tilde{p}(X_u=1 \vert x_{\neq u}) \right\vert \nonumber\\
& \leq &  \frac{2}{ p( x_{\neq u})}   \left\vert { p({X_u}=1,  x_{\neq u})- \tilde{p}({X}_u=1,  x_{\neq u})} \right\vert +  \left\vert {\tilde{p}(x_{\neq u}) -  p( x_{\neq u}) }\right\vert. 
\end{eqnarray}
Now we consider the bound in the following two cases:  1) if the two terms can be bounded by $L_p$ norm, we have $ \left\vert {\tilde{p}( x_{\neq u}) -  p( x_{\neq u}) } \right\vert \leq 2^{(n-1)(1-{1}/{p})}\epsilon_p $, $ \left\vert {\tilde{p}( X_u=1,  x_{\neq u}) -  p( X_u=1, x_{\neq u}) } \right\vert \leq  2^{(n-1)(1-{1}/{p})}\epsilon_p $. 2) for coherent bit-flip distance, by Lemma \ref{lemma_flip_error_bino_bound_sum}, we have  $\left\vert {\tilde{p}( x_{\neq u}) -  p( x_{\neq u}) } \right\vert \leq 2(1-\bar{\varrho}^{n-1}) $, $ \left\vert {\tilde{p}( X_u=1, x_{\neq u}) -  p( X_u=1, x_{\neq u}) } \right\vert \leq 2(1-\bar{\varrho}^{n})$. Additionally, for each case, we have $p(x_{\neq u})) \geq \sigma(-2\beta)^{d_2}$, after simple calculation, we can then obtain  Eq.~(\ref{eq_appendix_conditional_probability_bound}). For case 2), we have $\frac{8(1-\bar{\varrho}^{n})}{\sigma(-2\beta)^{d_2}} \leq \frac{8n\varrho}{\sigma(-2\beta)^{d_2}}$. 
\end{proof}

\begin{lemma}\label{lemma_appendix_potentail_2_norm}
Let $p(X), \tilde{p}(X), \mathcal D, \widetilde{\mathcal D}, q_u(X_{\neq u}), \tilde{q}_u(X_{\neq u})$ as the same setting with Lemma \ref{appendix_lemma_partial_potental_bound}. Suppose $q_u(X_{\neq u})$
contains at most $d_2$ nodes, for any $u\in [n]$. Given $M$ samples from distribution $\widetilde{\mathcal D}$, we can find a polynomial $\tilde{q}_u^*(X_{\neq u})$ which is an estimation of $\tilde{q}_u(X_{\neq u})$, by using Alphatron algorithm with probability at least $1-\delta$ such that
\begin{eqnarray}
\Vert \tilde{q}_u^* - \tilde{q}_u \Vert_2^2\leq \Ord{\frac{err}{2^{d_2}\sigma(-2\beta')^{d_2}(1-\tanh^2(\beta'))^2 }
}  
\end{eqnarray}
where $err = \epsilon'  +  \epsilon'  \sqrt[4]{\frac{\log( 1 /\delta)}{M}}+ \beta 2^{d_2/2}\sqrt{\frac{\log( 1 /\delta)}{M}}, \beta' =\beta +\epsilon',$ $\tilde q_u$ and $\tilde{q}^*_u$ are the  vector of coefficient of monomial $\tilde{q}_u(X_{\neq u})$ and  $\tilde{q}^*_u(X_{\neq u})$ respectively.   
\end{lemma}

\begin{proof}
We use the Alphatron algorithm as in Theorem \ref{theorem_alphatron} with the following setting.
Set  $X' = \left(\prod_{s\in S}x_S\right)_{S\subset \mathcal N_2(u)}$ and $Y'=X_u$.
Note that $\Vert X' \Vert_2 \leq 2^{d_2/2}$ as there are $2^{d_2}$ terms at most in $X'$, i.e., $B_1=2^{d_2/2}$.
Set the function $\mathcal U(x)=\tanh(x)$, which is a $1$-Lipschitz function, i.e., $L=1$.
Since $|q_u(X_{\neq u})|\leq \beta$, 
by using Parseval's theorem, we have
$ \Vert q_u  \Vert_2 \leq \beta $.
Note that $q_u$ corresponds to the vector $v$ and $B_2=\beta'$ in the Alphatron algorithm. 
By using Eq.~(\ref{eq_q_u_with_error_terms}), we have  
\begin{eqnarray}
\mathcal U (\tilde {q}_u(X_{\neq u}))= \tanh({q}_u(X_{\neq u})+ q_\epsilon(X_{\neq u}) ).  
\end{eqnarray}
From Lemma \ref{appendix_lemma_partial_potental_bound}, we have $\epsilon_1 \leq \epsilon' $, $\epsilon_b \leq \epsilon'^2.$
Under this setting, the Alphatron algorithm returns $\tilde{q}^*_u$ such that
\begin{eqnarray}
\mathbb E[(\tanh( \tilde{q}^*_u(X ))- \tanh(\tilde{q}_u( X) ))^2] 
\leq \Ord{\epsilon'   +  \epsilon'   \sqrt[4]{\frac{\log(1/\delta)}{M}}+ \beta2^{d_2/2}\sqrt{\frac{\log(1/\delta)}{M}}}: = \Ord{err}
\end{eqnarray}
with probability $1-\delta.$ This is from Theorem 6.8 in Ref.~\cite{bresler2019learning}. Since the derivative of tanh on $[-\beta',\beta']$ is lower bounded by $1-\tanh^2(\beta')$ and it is monotonically increasing, we obtain
\begin{eqnarray}
\mathbb E \left[\left(\tilde{q}^*_u(X) - \tilde{q}_u( X) \right)^2\right] =\Ord{\frac{err}{(1-\tanh^2(\beta'))^2} }  
\end{eqnarray}
By using similar proof in Lemma 6.1 in \cite{bresler2019learning}, we see $\tilde{p}(x_u) = \sigma(2\tilde{q}_u) \geq  \sigma(-2\beta').$ Iterating using the proof for $d_2$ nodes, we have $\tilde{p}(x_{\mathcal N(u)}) \geq \sigma(-2\beta')^{d_2}$. By using Lemma 6.9 in Ref.~\cite{bresler2019learning}, then we can obtain  $\sigma(-2\beta')^{d_2}=\frac{\delta}{2^{d_2}},$ with $\delta = 2^{d_2}(\sigma(-2\beta')^{d_2})$,  
\begin{eqnarray}
\Vert \tilde{q}^*_u - \tilde{q}_u \Vert_2^2\leq \Ord{\frac{err}{2^{d_2}\sigma(-2\beta')^{d_2}(1-\tanh^2(\beta'))^2 }
}.\label{eq_patial_q_2_norm_bound}
\end{eqnarray} 
\end{proof}

\subsection{Quantum state learning}

By employing an iterative application of the Alphatron algorithm to each visible node, we can estimate the coefficients of the potential polynomial of the induced Markov Random Field (MRF). Subsequently, we can recover the probability distribution using these parameters, ensuring bounded errors. When shifting the probability $\delta$ to $\delta' = \delta/n$ in Lemma \ref{lemma_appendix_potentail_2_norm}, the overall success probability remains at $\delta$ by the union bound. 
We can obtain $\tilde{q}^*$ by selecting the coefficient of a monomial $x_S$ equal to (arbitrarily chosen) the corresponding coefficient of $x_{S\setminus \{u\}}$ in $\tilde{q}^*_u$ for some $u$ within the set $S$. Finally, we have 
\begin{eqnarray}
\Vert \tilde{q}^* - \tilde{q} \Vert_2^2= \Ord{\frac{err(n)}{2^{d_2}(\sigma(-2\beta')^{d_2})(1-\tanh^2(\beta'))^2 }}, \label{eq_bound_potential_coefficient_bound}  
\end{eqnarray}
where $err(n) = \epsilon'   +  \epsilon'  \sqrt[4]{\frac{\log( n /\delta)}{M}}+ \beta 2^{d_2/2}\sqrt{\frac{\log( n /\delta)}{M}}.$

Based on the $L_2$ norm distance between the estimated and the true coefficients of the potential of the induced MRF, we show that the $L_1$ distance between the estimated magnitude and the actual magnitude of a state can be bounded as follows.

\begin{lemma}\label{theorem_parameter_leanring_main}
Given $M$ copies of an unknown $n$-qubit quantum state close to a quantum state based on an $(\alpha,\beta)$-nondegenerate LC-RBM with two-hop degree $d_2$. Given $M$ copies of this state, we can learn the distribution probability $\tilde{p}(x)$ of the unknown quantum state  by Alphatron algorithm  with probability at least $1-\delta$  such that 
\begin{eqnarray}
\sum_x\left\vert \tilde{p}(x) -\tilde{p}^*(x) \right\vert
=  \Ord{\frac{\sqrt{n} \sqrt{err(n)}}{\sigma(-2\beta')^{d_2/2}(1-\tanh^2(\beta'))  } + \left(\sqrt{n2^{d_2}} + n\right)\epsilon'},
\end{eqnarray}
where $\tilde{p}^*(x)$ is estimation of $\tilde{p}(x)$, where $\epsilon$ satisfies the bound in Lemma \ref{lemma_appendix_conditional_probability_bound},   $err(n) =  \epsilon'   +   \epsilon'   \sqrt[4]{\frac{\log( n /\delta)}{M}}+ \beta 2^{d_2/2} \sqrt{\frac{\log( n /\delta)}{M}}. $
\end{lemma}

\begin{proof}

After obtaining the coefficients of the induced MRF's potential, we can obtain the probability distribution of the quantum state. We now analyze the distance between the actual and the estimated probabilities.

First, we bound the distance between the actual and the estimated potential as follows 
\begin{eqnarray}
\left\vert  \tilde{q}(x) - \tilde{q}^*(x) \right\vert   = 
\sum_{I} \left\vert  \tilde{q}^*_I -{q}_I \right\vert x_I + n\epsilon(x) 
&  \leq  & 
\sum_{I} \left\vert \tilde{q}^*_I -{q}_I \right\vert  +  n\epsilon' \\   
& \leq &   \sqrt{n2^{d_2}}  \left\Vert  \tilde{q}^* -{q} \right\Vert_2 + n\epsilon' \nonumber\\
& \leq & \sqrt{n2^{d_2}}  \left\Vert  \tilde{q}^* -\tilde{q} +q_{\epsilon} \right\Vert_2 + n\epsilon'  \nonumber\\
& \leq & \sqrt{n2^{d_2}}  (\left\Vert  \tilde{q}^* -\tilde{q} \right\Vert_2 + \left\Vert q_{\epsilon} \right\Vert_2 )+  n\epsilon'  
\nonumber\\
& \leq & \sqrt{n2^{d_2}}  (\left\Vert  \tilde{q}^* -\tilde{q} \right\Vert_2 + \epsilon' )+  n\epsilon'  
\label{eq_apendix_potential_distance}
\end{eqnarray}
where $I \subset [n]$ with size at most $d_2+1$, then there are $n2^{d_2}$ monomials of $\tilde{q}^*(x)$ and ${q}(x)$, the right side of the first line is obtained using the triangle inequality and the second last inequality is according to the fact that $\Vert x \Vert_1 \leq \sqrt{n} \Vert x \Vert_2$ for a vector $x$ with dimension $n$.

Let the partition functions $\tilde{Z}:=\sum_x e^{\tilde{q}(x)}$, 
$\tilde{Z}^*:=\sum_x e^{\tilde{q}^*(x)}$.
First, assume $\tilde{Z}\geq \tilde{Z}^*$, 
then we have 

\begin{eqnarray}
\sum_x \left\vert \tilde{p}(x) -\tilde{p}^*(x) \right\vert
& = & \sum_x\left\vert \frac{e^{\tilde{q}(x)}}{\tilde{Z}}-\frac{e^{\tilde{q}^*(x)}}{\tilde{Z}^*}
\right\vert \nonumber\\
& = & \sum_x\left\vert \frac{e^{\tilde{q}(x)}}{\tilde{Z}}- \frac{e^{\tilde{q}^*(x)}}{\tilde{Z}} +\frac{e^{\tilde{q}^*(x)}}{\tilde{Z}} -\frac{e^{\tilde{q}^*(x)}}{\tilde{Z}^*}
\right\vert \nonumber\\
&\leq &  
\frac{1}{\tilde{Z}}\sum_x\left\vert {e^{\tilde{q}(x)}}- {e^{\tilde{q}^*(x)}} \right\vert +\left\vert \frac{1}{\tilde{Z}} -\frac{1}{\tilde{Z}^*}
\right\vert \sum_x e^{\tilde{q}^*(x)}, \label{eq_appendix_trace_distance}
\end{eqnarray}
where the inequality is obtained by using the triangle inequality. 
We now focus on the second term of Eq.~(\ref{eq_appendix_trace_distance})
\begin{eqnarray}
\left\vert \frac{1}{\tilde{Z}} -\frac{1}{\tilde{Z}^*}
\right\vert \sum_x e^{\tilde{q}^*(x)}
= \left\vert \frac{1}{\tilde{Z}} -\frac{1}{\tilde{Z}^*}
\right\vert {\tilde{Z}^*} = \left\vert \frac{\tilde{Z}^*-\tilde{Z}}{\tilde{Z}} 
\right\vert 
= 
\frac{1}{\tilde{Z}}\left\vert \sum_x (e^{\tilde{q}^*(x)}-  e^{\tilde{q}^*(x)}) \right\vert  \leq  \frac{1}{\tilde{Z}}\sum_x\left\vert {e^{\tilde{q}(x)}}- {e^{\tilde{q}^*(x)}} \right\vert,
\end{eqnarray}
where the inequality is obtained by using the triangle inequality. 
Combining the Eq.~(\ref{eq_appendix_trace_distance}), we have 
\begin{eqnarray}
\sum_x \left\vert \tilde{p}(x) -\tilde{p}^*(x) \right\vert  
& \leq &  
\frac{2}{\tilde{Z}}\sum_x\left\vert {e^{\tilde{q}(x)}}- {e^{\tilde{q}^*(x)}}  \right\vert  \nonumber\\
& \leq &  
\frac{2}{\tilde{Z}}\sum_x \max\{ {e^{\tilde{q}(x)}}, {e^{\tilde{q}^*(x)}} \}\left\vert {{\tilde{q}(x)}}- {{\tilde{q}^*(x)}}  \right\vert  \nonumber\\  
& \leq &  
\frac{2}{\tilde{Z}} \sqrt{n2^{d_2}} \left\Vert  \tilde{q}^* -\tilde{q} \right\Vert_2  \sum_x \max\{ {e^{\tilde{q}(x)}}, {e^{\tilde{q}^*(x)}} \}\nonumber\\ 
& \leq &  
\frac{2 \sum_x \max\{ {e^{\tilde{q}(x)}}, {e^{\tilde{q}^*(x)}} \}}{\frac{1}{2}(\tilde{Z} +\tilde{Z}^* ) } \left( \sqrt{n2^{d_2}} \left(\left\Vert  \tilde{q}^* -\tilde{q} \right\Vert_2 + \epsilon'\right) +  n\epsilon' \right) \nonumber\\
& \leq &  4 \sqrt{n2^{d_2}} \left(\left\Vert  \tilde{q}^* -\tilde{q} \right\Vert_2 +\epsilon'\right)+ 4n \epsilon'
\label{eq_bound_potential}  
\end{eqnarray}
where the right side of the second line is obtained by using the Cauchy inequality for a convex function $f=e^x$ which is $f'(b)\leq \frac{f(a)-f(b)}{a-b}\leq f'(a)$ for $a>b.$ 
The third inequality results from  Eq.~(\ref{eq_apendix_potential_distance}), and the last inequality is obtained by using the fact that $\tilde{Z}\geq \tilde{Z}^*$ and $\sum_x \max\{ e^{\tilde{q}(x)}, {e^{\tilde{q}^*(x)}} \} \leq \sum_x \left( e^{\tilde{q}(x)} + e^{\tilde{q}(x)} \right).$ Similarly, when $\tilde{Z}\leq \tilde{Z}^*$, we can obtain the same bound in Eq.~(\ref{eq_bound_potential}). 
Therefore, by using Eq.~(  \ref{eq_bound_potential_coefficient_bound}),  we have 
\begin{eqnarray}
\sum_x\left\vert \tilde{p}(x) -\tilde{p}^*(x) \right\vert
& \leq &   4 \sqrt{n2^{d_2}} \left(\left\Vert  \tilde{q}^* -\tilde{q} \right\Vert_2 +\epsilon'\right)+ 4n \epsilon'\nonumber\\
& = & \Ord{\frac{\sqrt{n} \sqrt{err(n)}}{\sigma(-2\beta')^{d_2/2}(1-\tanh^2(\beta'))  } + \left(\sqrt{n2^{d_2}} + n\right)\epsilon'
} 
\end{eqnarray}
where $err(n) = \epsilon'   +  \epsilon'  \sqrt[4]{\frac{\log( n /\delta)}{M}}+ \beta 2^{d_2/2} \sqrt{\frac{\log( n /\delta)}{M}}$  and $\beta'=\beta + \epsilon'.$
\end{proof}

Now we show that the magnitude of a quantum state in class $\mathcal C(\epsilon_p)\cup \mathcal C(\varrho)$ can be estimated if $\epsilon_p$ or $\epsilon_\varrho$ is small enough. 

\begin{theorem}[Quantum state learning]\label{thmLearnAll}
Suppose we are given $M$ copies of an unknown $n$-qubit quantum state in the union class of $\mathcal C(\epsilon_p)$ and  $ \mathcal C(\varrho)$ Defined in Def.~\ref{definition_quantum_classes_bit_flip_appendix} and \ref{definition_quantum_classes_appendix}, and the estimated two-hop neighbor's structure of the RBM representation be known. 
Let $\vartheta =\sigma(-2\beta)^{2d_2}(1-\tanh(\beta)^2)^3$, 
if the probability distribution  satisfies one of the following  constraints
\begin{eqnarray}
& (1).&  L_p\ \text{distance}:~  %\left(\sum_{x \in\{\pm 1\}^n} \vert \tilde p(x) - p(x)\vert^p \right)^{1/p} 
\epsilon_p =  \Ord{\frac{ \vartheta 
\epsilon_t^2}{\left(\sqrt{n2^{d_2}} + n\right)^2 2^{(n-1)(1-1/p)}}};\label{eq_main_theorem_qse_frbm_1}
\\
&(2). & \text{Bit flip probability~} : \varrho = \Ord{
\frac{\vartheta \epsilon_t^2}{n\left(\sqrt{n2^{d_2}} + n\right)^2}
}, \label{eq_main_theorem_qse_frbm_3}     
\end{eqnarray}
and 
\begin{eqnarray}
M= \Ord{\frac{\beta^2 2^{d_2} n^2\log(1/\delta)}{\vartheta (1-\tanh^2(\beta)) \epsilon_t^4}
} 
\end{eqnarray}
we can find a polynomial $\tilde{q}^*$  with probability $1-\delta$ such that
\begin{eqnarray}
\sum_x\left\vert \tilde{p}(x) -\tilde{p}^*(x) \right\vert
=  \Ord{\epsilon_t} 
\end{eqnarray}
where $ p(x) = \frac{e^{q(x)}}{Z}$.
\end{theorem}
\begin{proof}
By using Theorem \ref{theorem_parameter_leanring_main}, if we require 
\begin{eqnarray}
\sum_x\left\vert \tilde{p}(x) -\tilde{p}^*(x) \right\vert
\in  \Ord{\frac{\sqrt{n}  \sqrt{2\epsilon'   +  2\epsilon'   \sqrt[4]{\frac{\log( n /\delta)}{M}}+ \beta 2^{d_2/2} \sqrt{\frac{\log( n /\delta)}{M}}}}{\sigma(-2\beta')^{d_2/2}(1-\tanh^2(\beta'))  } +\left(\sqrt{n2^{d_2}} + n\right)\epsilon'
} =\Ord{\epsilon_t}
\end{eqnarray}
where $\beta'=\beta + \epsilon'.$
Let $\chi = \sigma(-2\beta')^{d_2/2}(1-\tanh^2(\beta'))$.
By using the fact that $\sqrt{a +b} \leq \sqrt{a}+\sqrt{b}$ for $a,b>0$, then each one of the following inequality should be satisfied
\begin{eqnarray}
\max \left\{ \frac{\sqrt{2n\epsilon'}}{\chi}, 
\frac{\sqrt{2n\epsilon'} \sqrt[8]{\frac{\log( n /\delta)}{M}}}{\chi},
\frac{\sqrt{n\beta 2^{d_2/2} } \sqrt[4]{\frac{\log( n /\delta)}{M}}}{\chi },
\left(\sqrt{n2^{d_2}} + n\right)\epsilon'
\right\}
=  \Ord{\epsilon_t}.
\end{eqnarray}
Using a similar approach as in the proof of Theorem \ref{theorem_appendix_conditional_probability},  we can conclude that  the above inequality holds if either Eq.(\ref{eq_main_theorem_qse_frbm_1}) or (\ref{eq_main_theorem_qse_frbm_3}) satisfied, and the number of samples satisfies
\begin{eqnarray}
M & \geq & \tOrd{\frac{2^{d_2}\beta^2 n^2\log(1/\delta)}{\sigma(-2\beta)^{2d_2}(1-\tanh^2(\beta))^4 \epsilon_t^4 }}.
\end{eqnarray}
\end{proof}

\subsection{Quantum state partial learning}

Next, we show that, under certain conditions, with just $O(\log n)$ copies of a quantum state (from the class in Theorem \ref{thmLearnAll}), we can learn the conditional probability of a subset of qubits relative to the remaining qubits with good precision.
Let $J \subset [n]$ and $\bar{J}=[n]\setminus J$. As discussed before, the conditional probability $p(x_J\vert x_{\bar J}) = p(x_J\vert x_{\mathcal N_2(J) })$, where $\mathcal{N}_2(J)$ is the visible node set composed of two-hop neighbors of each node in $J$ but excluding the nodes in $J$. Let $J_{\text{t}}= J\cup \mathcal{N}_2(J)$, and  $q(X_{J_{\text{t}}})$ to be a polynomial which is a part of polynomial $q(X)$, where each monomial in $q(X_{J_{t}})$ contain at least one node in $J$.
Naturally we also define $\bar{q}(X_{J_{t}})=q(X)-q(X_{J_{t}})$.
Denote the coefficient vector of $q(X_{J_{t}})$ as $q({J_{t}}).$ We have the following theorem.
\begin{theorem}[Quantum state partial learning]\label{theorem_appendix_conditional_probability} Let an unknown $n$-qubit quantum state as in Theorem \ref{thmLearnAll}, and the estimated two-hop neighborhood structure of the RBM representation be known.
Let $J \subseteq [n]$, set $J_{\rm t} = J\cup \mathcal N_2(J)$, and $\vartheta =\sigma(-2\beta)^{2d_2}(1-\tanh(\beta)^2)^3, \epsilon_c, \delta \in(0,1)$.
Suppose we are given $M_c$ copies of the qubits indexed by $J_{\rm tot}$ of the state.
If the following constraint (1) is satisfied for state $\ket{\psi}$ or  constraint (2) is satisfied for state $\ket{\Psi}$, 
\begin{eqnarray}
& (1).&  L_p\ \text{distance}:~  %\left(\sum_{x \in\{\pm 1\}^n} \vert \widetilde p(x) - p(x)\vert^p \right)^{{1}/{p}} 
\epsilon_p= 
\Ord{\frac{\vartheta  \epsilon_c^2}{\left( \sqrt{ 2^{d_2}\vert J \vert } + \vert J \vert \right)^2 2^{(n-1)(1-{1}/{p})}}};\label{eq_theorem_lcp_rbm_1}
\\
&(2). & \text{Bit flip probability~} : \varrho = \Ord{
\frac{\vartheta \epsilon_c^2}{\left( \sqrt{ 2^{d_2}\vert J \vert } + \vert J \vert \right)^2 n }
}, \label{eq_theorem_lcp_rbm_2}  \end{eqnarray}
and the number of samples satisfies 
\begin{eqnarray}
M_c=\Ord{\frac{\beta^2 2^{d_2}\vert J \vert^2\log(1/\delta)}{\vartheta (1-\tanh^2(\beta)) \epsilon_c^4}}, 
\end{eqnarray}
we can find a polynomial $q^*(X_{J_t})$ with probability $1-\delta$ such that 
\begin{eqnarray}
\vert  \tilde p(x_J \vert x_{\mathcal N_2(J)}) -  \tilde{p}^*(x_J \vert x_{{\mathcal N_2(J)}}) \vert 
=  \Ord{\epsilon_c},\label{eq_condition_p_distance_theorem}
\end{eqnarray}
where $p(x_J \vert x_{\bar{J}})= \frac{ \exp( q(x_{J_{t}})) }{ \sum_{x_J} \exp (q(x_{J_{t}}))}.$
\end{theorem}

\begin{proof}
The proof is similar to that of the Theorem \ref{thmLearnAll}. First, expanding the condition probability by using Bayes' law and the marginal probability, we have 
\begin{eqnarray}
p(x_J \vert x_{\bar{J}}) & = & \frac{ p(x_J, x_{\bar{J}})}{ \sum_{x_J} p(x_J, x_{\bar{J}})}
 =  \frac{ \exp( q(x_{J_{t}}) + \bar{q}(x_{J_{t}}))}{ \sum_{x_J} \exp (q(x_{J_{t}})+\bar{q}(x_{J_{t}}) )}
 =  \frac{ \exp( q(x_{J_{t}})) }{ \sum_{x_J} \exp (q(x_{J_{t}}))}
=  \frac{1 }{1+\sum_{k_J \neq x_J} \exp (q(k_{J_{t}})-q(x_{J_{t}}))},\nonumber\\
\end{eqnarray}
where $x_J, k_J \in \{\pm 1 \}^{\vert J \vert }.$  

Then we bound the difference between the estimated conditional probability and the real conditional probability as follows
\begin{eqnarray}
& &  \vert  \tilde p(x_J \vert x_{\bar{J}}) -  \tilde{p}^*(x_J \vert x_{\bar{J}}) \vert \nonumber\\
& = & \left\vert \frac{ 1 }{1+ \sum_{k_J\neq x_J} \exp (\tilde q(k_{J_{t}})-\tilde q(x_{J_{t}}))}- \frac{1}{ 1+ \sum_{k_J\neq x_J} \exp \left( \tilde{q}^*(k_{J_{t}}) -\tilde{q}^*(x_{J_{t}}) \right)        } \right\vert 
\nonumber\\
&=& \left\vert \frac{ \sum_{k_J\neq x_J} \exp \left( \tilde{q}^*(k_{J_{t}}) -\tilde{q}^*(x_{J_{t}}) \right) - \exp (\tilde q(k_{J_{t}})-\tilde q(x_{J_{t}})) }{ \left(1 + \sum_{k_J\neq x_J} \exp (\tilde q(k_{J_{t}})-\tilde q(x_{J_{t}}))\right)\left(1 + \sum_{k_J\neq x_J} \exp \left( \tilde{q}^*(k_{J_{t}}) -\tilde{q}^*(x_{J_{t}}) \right)  \right)}
\right\vert 
\nonumber\\
&\leq &  \frac{ \sum_{k_J\neq x_J} \left\vert  \exp \left( \tilde{q}^*(k_{J_{t}}) -\tilde{q}^*(x_{J_{t}}) \right) -  \exp (\tilde q(k_{J_{t}})-\tilde q(x_{J_{t}}))  \right\vert }{ \left(1+\sum_{k_J\neq x_J} \exp (\tilde q(k_{J_{t}})-\tilde q(x_{J_{t}}))\right)\left(1 + \sum_{k_J\neq x_J} \exp \left( \tilde{q}^*(k_{J_{t}}) -\tilde{q}^*(x_{J_{t}}) \right)  \right)}\label{eq_dif_condition_pr_prove_1}     
\\
&\leq &  \frac{ \sum_{k_J\neq x_J} \max \{ \exp \left( \tilde{q}^*(k_{J_{t}}) -\tilde{q}^*(x_{J_{t}}) \right), \exp (\tilde q(k_{J_{t}})-\tilde q(x_{J_{t}})) 
\}\left\vert   \left( \tilde{q}^*(k_{J_{t}}) -\tilde{q}^*(x_{J_{t}}) \right) -   (\tilde q(k_{J_{t}})-\tilde q(x_{J_{t}}))  \right\vert }{ \left(1+ \sum_{k_J\neq x_J} \exp (\tilde q(k_{J_{t}})-\tilde q(x_{J_{t}}))\right)\left(1+\sum_{k_J\neq x_J} \exp \left( \tilde{q}^*(k_{J_{t}}) -\tilde{q}^*(x_{J_{t}}) \right)  \right)} \label{eq_dif_condition_pr_prove_2} 
\\ &\leq &  \frac{ \sum_{k_J\neq x_J} \max \{ \exp \left( \tilde{q}^*(k_{J_{t}}) -\tilde{q}^*(x_{J_{t}}) \right), \exp (\tilde q(k_{J_{t}})-\tilde q(x_{J_{t}})) \}
\left( \left\vert    \tilde{q}^*(k_{J_{t}}) -   \tilde q(k_{J_{t}}) \right\vert +  \left\vert   \tilde{q}^*(x_{J_{t}})  -\tilde q(x_{J_{t}})  \right\vert \right)
}{ \left(1+ \sum_{k_J\neq x_J} \exp (\tilde q(k_{J_{t}})-\tilde q(x_{J_{t}}))\right)\left(1+\sum_{k_J\neq x_J} \exp \left( \tilde{q}^*(k_{J_{t}}) -\tilde{q}^*(x_{J_{t}}) \right)  \right)}, \label{eq_dif_condition_pr_prove_2*}    
\end{eqnarray}
where the Eq.~(\ref{eq_dif_condition_pr_prove_1}) and (\ref{eq_dif_condition_pr_prove_2*})  are obtained by using the triangle inequality, Eq.~(\ref{eq_dif_condition_pr_prove_2}) we use the property of  a convex function $f=e^x$ which is $f'(b)\leq \frac{f(a)-f(b)}{a-b}\leq f'(a)$ for $a>b.$ By using Eq.~(\ref{eq_q_u_with_error_terms}) and Lemma \ref{appendix_lemma_partial_potental_bound}, we have 
\begin{eqnarray}
\left\vert    \tilde{q}^*(k_{J_t}) -   \tilde q(k_{J_t}) \right\vert \leq \left\vert    \tilde{q}^*(k_{J_t}) -    q(k_{J_t})  \right\vert + \vert q_\epsilon(k_{J_t})\vert
& \leq & \sum_{I\subset J_t } \left\vert  \tilde{q}^*_I - q_I  \right\vert + \vert J \vert \epsilon' \nonumber\\
& \leq &
\sqrt{\vert J\vert  2^{d_2}}  \left\Vert  \tilde{q}^*(J_t) -{q}(J_t) \right\Vert_2 +  \vert J\vert \epsilon' \nonumber\\ 
& \leq &  \sqrt{\vert J\vert  2^{d_2}}  \left\Vert  \tilde{q}^*(J_t) -\tilde{q}(J_t) +q_{\epsilon}(J_t)\right\Vert_2 + \vert J\vert \epsilon' \nonumber\\
& \leq &  \sqrt{\vert J\vert  2^{d_2}} \left( \left\Vert  \tilde{q}^*(J_t) -\tilde{q}(J_t) \Vert_2 +  \Vert q_{\epsilon}(J_t)\right\Vert_2 \right)  + \vert J\vert \epsilon'\nonumber\\  
& \leq &  \sqrt{\vert J\vert  2^{d_2}} 
\left( \Vert  \tilde{q}^*(J_t) -\tilde{q}(J_t) \Vert_2 +  \epsilon' \right)  + \vert J\vert \epsilon'.
\label{eq_condition_bound_proof_3}
\end{eqnarray}
for any $k_J\in \{\pm 1\}^{\vert J \vert},$ where the first inequality is obtained by using the triangle equality and Eq.~(\ref{eq_q_u_with_error_terms}), and the second inequality is according to the fact that there are at most $\vert J \vert 2^{d_2}$ monomials in polynomial $q(X_{J_t})$ and  $\Vert x \Vert_1 \leq \sqrt{n} \Vert x \Vert_2$ for a vector $x$ with dimension $n$. For the last two inequalities, we use triangle inequality and the fact that $\Vert q_{\epsilon}(J_t)\Vert_2 \leq \vert q_{\epsilon}(J_t)\vert\leq \epsilon'.$

Additionally, we divided the sum of $\max \{\cdot\}$ in Eq.~(\ref{eq_dif_condition_pr_prove_2*}) into two parts, one part only containing terms from learning results (with $*$ sign), and the other part containing only the terms from the real distribution. Then we see 
\begin{eqnarray}
\frac{ \sum_{k_J\neq x_J} \max \{ \exp \left( \tilde{q}^*(k_{J_t}) -\tilde{q}^*(x_{J_t}) \right), \exp (\tilde q(k_{J_t})-\tilde q(x_{J_t})) 
\} }{ \left(1+ \sum_{k_J\neq x_J} \exp (\tilde q(k_{J_t})-\tilde q(x_{J_t}))\right)\left(1+ \sum_{k_J\neq x_J} \exp \left( \tilde{q}^*(k_{J_t}) -\tilde{q}^*(x_{J_t}) \right)  \right)}  <2. \label{eq_condition_bound_proof_4}
\end{eqnarray}
Combining  Eq.~(\ref{eq_dif_condition_pr_prove_2*}), Eq.~(\ref{eq_condition_bound_proof_3}), Eq.~(\ref{eq_condition_bound_proof_4}), we can obtain
\begin{eqnarray}     
\vert  \tilde p(x_J \vert x_{\bar{J}}) -  \tilde{p}^*(x_J \vert x_{\bar{J}}) \vert     
\leq   4  \left( \sqrt{\vert J \vert 2^{d_2}}  \left( \left\Vert  \tilde{q}^*(J_t) -\tilde q(J_t)\right\Vert_2 +\epsilon' \right)  +    \vert J \vert \epsilon' \right).\label{eq_dif_condition_pr_prove_5}   
\end{eqnarray}

The polynomial $q(X_{J_t})$ can be obtained by performing the Alphatron algorithm for all nodes $u\in J$. Let $\chi = \sigma(-2\beta')^{d_2/2}(1-\tanh^2(\beta')$.
By using Lemma \ref{lemma_appendix_potentail_2_norm}, we have
\begin{eqnarray}
\vert  \tilde p(x_J \vert x_{\bar{J}}) -  \tilde{p}^*(x_J \vert x_{\bar{J}}) \vert  = \Ord{\frac{  \sqrt{\vert J \vert } 
\sqrt{err(\vert J \vert )}}{\chi } + \left( \sqrt{\vert J \vert 2^{d_2}}  + \vert J \vert \right)
\epsilon'}, 
\end{eqnarray}
where $err(\vert J \vert ) = \epsilon'   +  \epsilon'   \sqrt[4]{\frac{\log( \vert J \vert  /\delta)}{M}}+ \beta 2^{d_2/2}\sqrt{\frac{\log( \vert J \vert /\delta)}{M}}.$

By using the fact that $\sqrt{a +b} \leq \sqrt{a}+\sqrt{b}$ for $a,b>0$, then the following inequality should be satisfied
\begin{eqnarray}
& &\max\left\{
\frac{\sqrt{\epsilon' \vert J \vert }}{\chi },
\frac{\sqrt{\epsilon' \vert J \vert} \sqrt[8]{\frac{\log( \vert J \vert /\delta)}{M_c}}}{\chi },
\frac{\sqrt{\beta 2^{d_2/2} \vert J \vert} \sqrt[4]{\frac{\log( \vert J \vert /\delta)}{M_c}}}{\chi}, 
\left(\sqrt{\vert J \vert 2^{d_2}}  +\vert J \vert \right) \epsilon' \right\} \nonumber\\
& = & \max \left\{
\frac{\sqrt{\epsilon' \vert J \vert }}{\chi },
\frac{\sqrt{\beta 2^{d_2/2}\vert J \vert} \sqrt[4]{\frac{\log( \vert J \vert /\delta)}{M_c}}}{\chi}, 
\left(\sqrt{\vert J \vert 2^{d_2}}  +\vert J \vert \right) \epsilon' \right\}\nonumber\\
& \leq & \max \left\{
\frac{\sqrt{\epsilon'  }\left(\sqrt{\vert J \vert 2^{d_2}}  +\vert J \vert \right) }{\chi },
\frac{\sqrt{\beta 2^{d_2/2} \vert J \vert} \sqrt[4]{\frac{\log( \vert J \vert /\delta)}{M_c}}}{\chi} \right\}\label{eq_apendix_condition_P_max}
\end{eqnarray}
For the first term in Eq.~(\ref{eq_apendix_condition_P_max}), as  $\epsilon'= \frac{\epsilon}{1-\left(\tanh(\beta)+\epsilon\right)^2},$
we have
\begin{eqnarray}
\frac{{\epsilon' \left(\sqrt{\vert J \vert 2^{d_2}}  +\vert J \vert \right)^2}}{\chi^2 }\leq    \frac{{\epsilon \left(\sqrt{\vert J \vert 2^{d_2}}  +\vert J \vert \right)^2}}{\chi^2 \left(1-\left(\tanh(\beta)+\epsilon\right)^2\right)}
\end{eqnarray}
If we restrict the conditional probability distance as  Eq.~(\ref{eq_condition_p_distance_theorem}), combing with Lemmas \ref{appendix_lemma_partial_potental_bound} and \ref{lemma_appendix_conditional_probability_bound},  Eq.~(\ref{eq_theorem_lcp_rbm_1}) and Eq.~(\ref{eq_theorem_lcp_rbm_2}) should be satisfied.

For the second term of Eq.~(\ref{eq_apendix_condition_P_max}), the constraint in Eq.~(\ref{eq_condition_p_distance_theorem}) is satisfied if 
\begin{eqnarray}
M_c  =  \tOrd{\frac{\beta^2 2^{d_2}\vert J \vert ^2\log(1/\delta)}{\chi^4 \epsilon_c^4 } }. 
\end{eqnarray}
\end{proof}

We see if $d_2 =\Ord{\log\log n}$, $\vert J \vert = \log n ,$ it requires $M_c\in \Ord{poly (\log n)}$ samples  in terms the number of qubits.

\end{document}